\documentclass[pra,floatfix,twocolumn,amsfonts,amsmath,amssymb,superscriptaddress,nofootinbib,reprint]{revtex4-1}
\pdfoutput=1

\usepackage{amsthm}
\usepackage{mathtools}
\usepackage{bbm}
\usepackage{bm}

\usepackage{graphicx}
\usepackage{tikz}
\usetikzlibrary{arrows}

\newtheorem{theorem}{Theorem}
\newtheorem{proposition}[theorem]{Proposition}
\newtheorem{conjecture}[theorem]{Conjecture}

\newtheorem{fact}[theorem]{Fact}

\theoremstyle{definition}
\newtheorem{definition}[theorem]{Definition}

\def\ket#1{\mathinner{|{#1}\rangle}}

{\catcode`\|=\active 
  \gdef\Braket#1{\left<\mathcode`\|"8000\let|\BraVert {#1}\right>}}
\def\BraVert{\egroup\,\mid@vertical\,\bgroup}
\newcommand{\oprod}[2]{| #1 \rangle\langle #2 |}
\newcommand{\iprod}[2]{\langle #1 | #2 \rangle}

\DeclareMathOperator{\Tr}{Tr}
\DeclareMathOperator{\cl}{cl}
\DeclareMathOperator{\conv}{conv}
\newcommand{\M}{\mathcal{M}}
\newcommand{\E}{\mathcal{E}}
\newcommand{\id}{\mathbbm{1}}
\newcommand{\dagg}{\dagger}
\newcommand{\vect}[1]{\bm{#1}}
\newcommand{\vprod}[2]{\vect{#1}\cdot\vect{#2}}

\usepackage[utf8]{inputenc}
\usepackage[T1]{fontenc}
\usepackage[colorlinks]{hyperref}

\begin{document}

\title{Noise and Disturbance of Qubit Measurements: \\ An Information-Theoretic Characterisation}

\author{Alastair A. Abbott}
\email{alastair.abbott@neel.cnrs.fr}
\affiliation{Institut N\'{e}el, CNRS and Universit\'{e} Grenoble Alpes, 38042 Grenoble Cedex 9, France}

\author{Cyril Branciard}
% \email{cyril.branciard@neel.cnrs.fr}
\affiliation{Institut N\'{e}el, CNRS and Universit\'{e} Grenoble Alpes, 38042 Grenoble Cedex 9, France}

\date{\today}

\begin{abstract}
Information-theoretic definitions for the noise associated with a quantum measurement and the corresponding disturbance to the state of the system have recently been introduced~[F. Buscemi \emph{et~al.}, Phys. Rev. Lett. \textbf{112}, 050401 (2014)].
These definitions are invariant under relabelling of measurement outcomes, and lend themselves readily to the formulation of state-independent uncertainty relations both for the joint estimate of observables (noise-noise relations) and the noise-disturbance tradeoff.
Here we derive such relations for incompatible qubit observables, which we prove to be tight in the case of joint estimates, and present progress towards fully characterising the noise-disturbance tradeoff.
In doing so, we show that the set of obtainable noise-noise values for such observables is convex, whereas the conjectured form for the set of obtainable noise-disturbance values is not.
Furthermore, projective measurements are not optimal with respect to the joint-measurement noise or noise-disturbance tradeoffs. 
Interestingly, it seems that four-outcome measurements are needed in the former case, whereas three-outcome measurements are optimal in the latter.
\end{abstract}

% \pacs{}
% \keywords{}

%%%%%%%%%%%%%%%%%%%%%%%%%%%%%%%%%%%%%%%%%%
	
\maketitle

%%%%%%%%%%%%%%%%%%%%%%%%%%%%%%%%%%%%%%%%%%

\section{Introduction}

Heisenberg's uncertainty principle is one of the defining nonclassical features of quantum mechanics, and expresses one of the fundamental physical consequences of the noncommutativity of quantum observables.
Informally, the principle states that the measurement of one quantum observable (such as the position of a particle, $x$) introduces an irreversible disturbance into any complementary observable of the system (such as the particle's momentum, $p$), thus rendering it impossible to simultaneously measure, with arbitrary precision, the values of incompatible observable quantities.

Heisenberg's original presentation of the uncertainty principle, exhibited in his microscope \emph{Gedankenexperiment}~\cite{Heisenberg:1927zh}, was rather informal, and despite the evident physical importance of the principle it was a long time before it was rigorously formalised.
Instead, subsequent theoretical work on the incompatibility of quantum observables focused on the inability to produce states with sharply defined values associated with noncommuting observables.
These results are typically expressed in the form of uncertainty relations for the standard deviations of such observables -- such as Kennard's well known relation~\cite{Kennard:1927oa} $\Delta x\Delta p \ge \frac{\hbar}{2}$ -- and express a subtly different, although related, physical consequence of noncommutativity.
To avoid confusion, we will call such relations \emph{preparation uncertainty relations}.

It is only much more recently that, with the help of a more modern theory of quantum measurement~\cite{Kraus:1983aa}, it has become possible to more rigorously quantify the noise and disturbance of a measurement, e.g.\ by defining noise and disturbance measures based on the root-mean-square distance between target observables and the measurement made~\cite{Ozawa:2003fh} or by quantifying the distance between their output distributions~\cite{Werner:2004aa}.
This has allowed Heisenberg's uncertainty principle to be formalised in terms of \emph{measurement uncertainty relations} between such measures of noise and disturbance, although there still remains debate as to which measure is the most appropriate~\cite{Branciard:2013hb,Busch:2014ts,Dressel:2014fx,Hall:2004eq,Ozawa:2003fh,Werner:2004aa}.
In fact, one may distinguish further two forms of measurement uncertainty relations expressing the incompatibility of such  measurements~\cite{Busch:2014ts}: 
\emph{noise-noise relations} for joint measurements, expressing the tradeoff in precision with which two complementary observables can be simultaneously measured;
and \emph{noise-disturbance relations}, expressing the tradeoff between the precision of a measurement and the subsequent disturbance to the state with respect to a complementary observable. 

Perhaps motivated by the success of entropic (preparation) uncertainty relations~\cite{Coles:2015ef}, which use entropy rather than the standard deviation to measure the uncertainty associated with an observable for a given state, a recent proposal by Buscemi \emph{et al.}~\cite{Buscemi:2014aa}\ set out a new approach to quantifying the noise and disturbance associated with a measurement based on information-theoretic concepts.
This approach, in contrast to those mentioned above, uses the information gained and lost during measurement to provide intuitive measures of noise and disturbance; that is, it looks at the correlations between input states and measurement outcomes, using the notion of conditional entropy to quantify them.
As for entropic uncertainty relations, this approach is invariant under the relabelling of measurement outcomes and, furthermore, provides measures of noise and disturbance that are state-independent: they depend only on the complementary observables in question and the measurement performed.

More recently, several alternative information-theoretic approaches to defining noise and disturbance have been proposed.
Perhaps most notably, Ref.~\cite{Barchielli:2016aa} defines them in terms of the relative entropy between the distributions associated with the target observables and the measurement made.
This approach differs conceptually from that of Buscemi \emph{et al.}, which is instead based on the uncertainty in the post-measurement distribution conditioned on the \emph{pre}-measurement distribution, and is more in line with the approach of Ref.~\cite{Barnum:2000aa}.
Another proposal~\cite{Coles:2015aa}, albeit in a slightly different operational setting, combines these approaches, using the conditional entropy to define the noise and the relative entropy to define the disturbance.
Various other related information-theoretic~\cite{Baek:2016aa} and operational~\cite{Schwonnek:2016aa} approaches have also been recently investigated, emphasising the subtleties of the problem, but we will not discuss these further as we aim to tackle specific questions within the formalism of Buscemi \emph{et al.}~\cite{Buscemi:2014aa}.

In proposing this approach, the authors proved a state-independent measurement uncertainty relation that is valid for arbitrary observables in any finite Hilbert space~\cite{Buscemi:2014aa}. 
However, as is the case with similar preparation uncertainty relations, the result is far from tight in general.
It is thus of interest to look at simpler systems to find tight relations and fully understand the noise-noise and noise-disturbance tradeoffs.
The simplest nontrivial quantum system one can envisage is, of course, the qubit, and in a subsequent paper an apparently tight noise-disturbance relation for orthogonal qubit observables was proposed and tested experimentally~\cite{Sulyok:2015aa}.
Unfortunately, as we will discuss, the proof of this relation was incorrect, thus casting doubt on its validity;
indeed, we will show that it is incorrect in general, although it can be shown to hold in some particular cases.

In this paper, we revisit the qubit scenario, looking not only at noise-disturbance relations, but also at noise-noise relations for joint measurements.
We completely characterise the joint-measurement scenario for arbitrary qubit observables, showing that the set of obtainable noise-noise values is convex and that it seems four-outcomes measurements are required to saturate the tradeoff.
On the other hand, we provide evidence that the set of obtainable noise-disturbance points is non-convex, and that three-outcome measurements are both necessary and sufficient to saturate the tradeoff.
Finally, we prove that measurements made using ``L\"{u}ders instruments'', a natural class of instruments in which the state is updated according to the so-called ``square-root dynamics'', are not optimal and that in fact they satisfy the (more restrictive) relation originally given in Ref.~\cite{Sulyok:2015aa}.
Thus, non-trivial corrections are needed to perform optimal measurements with respect to the noise-disturbance tradeoff.

\section{Theoretical framework: Entropic definitions of noise and disturbance}

Let us first outline the information-theoretic framework for quantifying noise and disturbance that we shall use, and which was first presented in~\cite{Buscemi:2014aa}.

We shall consider two (for simplicity, non-degenerate) observables $A$ and $B$ on a finite dimensional Hilbert space with respective (normalised) eigenstates $\{\ket{a}\}_a$ and $\{\ket{b}\}_b$, where $a$ and $b$ label the respective eigenvalues (their numerical values are irrelevant).
According to quantum theory, the measurement device $\M$, with measurement outcomes labelled by $m$, is represented in the most general way possible as a quantum instrument~\cite{Davies:1970aa}.
Let us recall the definition of a quantum instrument.
\begin{definition}
	A \emph{quantum instrument} $\M$ is a collection $\{\M_m\}_m$ of completely positive (CP) trace-non-increasing maps $\M_m$ such that the map\footnote{This slight abuse of notation is generally unambiguous and proves convenient.} $\M=\sum_m\M_m$ is a completely positive trace-preserving (CPTP) map, i.e., $\Tr[\M(\rho)]=\Tr[\rho]$ for all Hermitian $\rho$.	
	The probability of obtaining outcome $m$ when measuring $\M$ on any (normalised) state $\rho$ is $\Tr[\M_m(\rho)]$, and the post-measurement state is $\frac{\M_m(\rho)}{\Tr[\M_m(\rho)]}$.
	
Every instrument $\M=\{\M_m\}_m$ uniquely defines a positive-operator valued measure (POVM) $M=\{M_m\}_m$ whose elements\footnote{The POVM elements $M_m$ can be obtained from the (non-unique) Kraus operators $\{K_{m,i}\}_i$ in the operator-sum representation of $\M_m$ as $\M_m(\rho)=\sum_i K_{m,i} \rho K_{m,i}^\dagg$. Specifically, one has $M_m=\sum_i K_{m,i}^\dagg K_{m,i}$.} $M_m$ are Hermitian positive semidefinite operators satisfying $\sum_m M_m=\id$ (where $\id$ is the identity operator) and $\Tr[M_m \rho]=\Tr[\M_m(\rho)]$ for all $\rho$.
This POVM determines only the probability of each measurement outcome, ignoring the post-measurement state.
\end{definition}

Let us first consider the noise of $\M$ with respect to $A$, $N(\M,A)$.
Imagine an experiment in which the eigenstates $\ket{a}$ of $A$ are prepared with equal probability and measured by $\M$.
The correlation between the eigenvalue $a$ corresponding to the state prepared and the outcome $m$ measured, which will be used to define the noise, is characterised by the joint probability distribution 
\begin{equation}\label{eqn:jointDistMA}
	p(m,a)=p(a)p(m|a)=\frac{1}{d}\,p(m|a),
\end{equation}
where $d$ is the Hilbert space dimension, and 
\begin{equation}\label{eqn:jointDistMAcond}
	p(m|a)=\Tr[\M_m(\oprod{a}{a})]=\Tr[M_m\,\oprod{a}{a}].
\end{equation}

We denote the classical random variables associated with $a$ and $m$ by $\mathbb{A}$ and $\mathbb{M}$, respectively.
This scenario is depicted schematically in Fig.~\ref{fig:NDSchematic}(a).

Recall that the Shannon entropy $H(\mathbb{X})$ of a random variable $\mathbb{X}$ distributed according to $p(x)$ is defined as
\begin{equation}\label{eqn:ShannonEntropy}
	H(\mathbb{X})=-\sum_x p(x) \log p(x),
\end{equation}
where the logarithms are taken in base 2 (as are all subsequent ones appearing in this paper).
\begin{definition}
	The \emph{noise} of $\M$ for a measurement of $A$ is $N(\M,A)=H(\mathbb{A}|\mathbb{M})$, where $H(\mathbb{A}|\mathbb{M})=H(\mathbb{A},\mathbb{M})-H(\mathbb{M})$ is the conditional entropy of $\mathbb{A}$ given $\mathbb{M}$ and can be calculated directly from the joint distribution~\eqref{eqn:jointDistMA} and the marginal distribution $p(m)=\sum_a p(m,a)$.
\end{definition}

This definition of noise thus quantifies the uncertainty as to which eigenstate was prepared, given the measurement outcome $m$ of $\M$.

By writing the conditional entropy explicitly in an alternative, equivalent form as
\begin{align}
	H(\mathbb{A}|\mathbb{M}) &= \sum_m p(m)H(\mathbb{A}|\mathbb{M}=m)\notag\\
		&= -\sum_m p(m)\sum_a p(a|m)\log p(a|m)
\end{align}
it is possible to express the noise in terms of the entropies $H(A|\rho_m)$ of the quantum observable $A$ for a set of states $\{\rho_m\}_m$, where $H(A|\rho)$ is defined as 
\begin{equation}\label{eqn:quantumEntropy}
H(A|\rho)=-\sum_a \Tr\big[\oprod{a}{a}\,\rho\big]\log\Tr\big[\oprod{a}{a}\,\rho\big].
\end{equation}
Specifically, by explicit calculation from the joint distribution $p(m,a)$, we have 
\begin{equation}\label{eqn:mMarginal}
	p(m)=\frac{1}{d}\Tr[M_m]
\end{equation}
and
\begin{equation}
	p(a|m) = \Tr\left[ \oprod{a}{a} \frac{M_m}{\Tr[M_m]} \right].
\end{equation}
Noting that, for all $m$, $\rho_m=\frac{M_m}{\Tr[M_m]}$ is a semidefinite positive trace-1 operator and thus defines a valid quantum state,
we see that $H(\mathbb{A}|\mathbb{M}=m)=H(A|\rho_m)$ and thus
\begin{equation}\label{eqn:NoiseConvEntropyRep}
	N(\M,A)=\sum_m p(m) H(A|\rho_m).
\end{equation}
This result was derived in the supplemental material of Ref.~\cite{Buscemi:2014aa} via a substantially more complicated argument, and in the Appendix we discuss an operational interpretation of this result and its relation to the approach of~\cite{Buscemi:2014aa}.
Note finally that the noise depends only on the POVM $M$, and not the full description of the instrument $\M$.

\begin{figure*}[ht]
	\begin{center}
	\begin{tabular}{ccc}
	\includegraphics[scale=0.95]{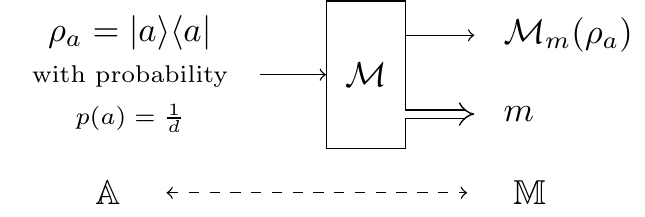}
	&
	\qquad
	&
	\includegraphics[scale=0.95]{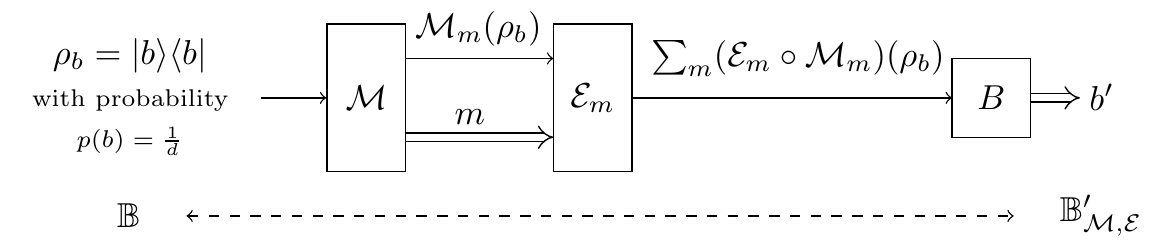}
	\\
	(a) && (b)\\
	\end{tabular}
	\end{center}
	\caption{Schematics of the scenarios used in the information-theoretic definitions of (a) noise, $N(\M,A)$, and (b) disturbance, $D(\M,B)$. The eigenstates $\ket{a}$ of $A$ (or $B$, for disturbance) are prepared with equal probability, before being measured by $\M$, producing outcome $m$ and transforming the state according to $\M_m$. In (b), a correction $\mathcal{E}_m$ is then applied and a further projective measurement of $B$ is performed generating the outcome $b'$, which is used to determine the disturbance.}
	\label{fig:NDSchematic}
\end{figure*}

The disturbance $D(\M,B)$ is defined with respect to an analogous experiment where this time eigenstates $\ket{b}$ of $B$ are prepared with equal probability, and one looks at the uncertainty in $B$ following the measurement.
This is quantified by the correlation between $b$ and the outcome $b'$ of a further projective measurement of $B$ following $\M$.
Since the definition is intended to quantify only the irreversible loss of information due to $\M$, a correction $\mathcal{E}_m$ may be performed prior to this subsequent measurement, where $\mathcal{E}_m$ is a CPTP map which may depend on the measurement outcome $m$.
This scenario is characterised by the joint probability distribution
\begin{equation}\label{eqn:jointDistBB}
	p(b',b)=p(b)p(b'|b)=\frac{1}{d}\,p(b'|b),
\end{equation}
where $p(b'|b)$ is given by the Born rule as
\begin{align}\label{eqn:jointDistBBcond}
	p(b'|b) 
		& =\Tr\left[ \sum_m(\mathcal{E}_m\circ\M_m)(\oprod{b}{b})\cdot \oprod{b'}{b'} \right].
\end{align}
We denote the random variables associated with $b$ and $b'$ by $\mathbb{B}$ and $\mathbb{B}_{\M,\E}'$, respectively.
This scenario is depicted in Fig.~\ref{fig:NDSchematic}(b).
\sloppy
\begin{definition}
	Let $\mathcal{E}=\{\mathcal{E}_m\}_m$ be a correction procedure. 
	The \emph{$\E$-disturbance} due to $\M$ on any subsequent measurement of $B$ is $D_\E(\M,B)=H(\mathbb{B}|\mathbb{B}_{\M,\E}')$, where the conditional entropy $H(\mathbb{B}|\mathbb{B}_{\M,\E}')$ is calculated from Eq.~\eqref{eqn:jointDistBB}.
	The \emph{disturbance} is then defined as $D(\M,B)=\min_{\mathcal{E}}D_\E(\M,B)$, where the minimisation is taken over all correction procedures $\mathcal{E}$.
\end{definition}
\fussy

This definition of disturbance thus quantifies the uncertainty as to which eigenstate was prepared, given the measurement outcome $b'$ of $B$ on the state after the measurement $\M$ and the optimal correction procedure $\mathcal{E}$.
Contrary to the case of noise (see Eq.~\eqref{eqn:NoiseConvEntropyRep}), there is no simple, general expression for the disturbance (although in some specific cases it is possible to calculate it more explicitly, cf.\ Appendix).
As we will see, this contributes to making the characterisation of the noise-disturbance tradeoff more complicated than it is for the noise-noise tradeoff.

We briefly note that these definitions of noise and disturbance do not generalise readily to infinite dimensional systems due the assumption that the eigenstates of the observables in question are prepared uniformly at random.
Although it is possible to modify the definitions in an attempt to address this, such modifications (e.g., those discussed in Ref.~\cite{Buscemi:2014aa} to accommodate continuous observables) lack much of the operational appeal of the above definitions for discrete systems.

\section{Measurement uncertainty relations}

\subsection{General case}

Using these notions of noise and disturbance, Ref.~\cite{Buscemi:2014aa} proved that, for arbitrary observables $A$ and $B$ in finite dimensional Hilbert spaces, both the noise-noise (joint-measurement) relation 
\begin{equation}\label{eqn:MUNN}
	N(\M,A) + N(\M,B)\ge -\log\,\max_{a,b}|\iprod{a}{b}|^2,
\end{equation}
and the noise-disturbance relation
\begin{equation}\label{eqn:MUND}
	N(\M,A) + D(\M,B)\ge -\log\,\max_{a,b}|\iprod{a}{b}|^2,
\end{equation}
hold.
That these relations bear a clear resemblance to the well-known Maassen and Uffink entropic preparation uncertainty relation~\cite{Maassen:1988vl} is no coincidence.
Indeed, their derivation relied on results (cf.\ Propositions~\ref{prop:RNNDef} and~\ref{prop:NN-ND-relation} below) providing bounds for both the noise and disturbance in terms of the entropic uncertainties $H(A|\rho)$ and $H(B|\rho)$ for the observables $A$ and $B$ and some state $\rho$, to which the state-independent Maassen and Uffink relation could be applied~\cite{Buscemi:2014aa}.

However, just like Maassen and Uffink's uncertainty relation, relations~\eqref{eqn:MUNN} and~\eqref{eqn:MUND} are not tight in general.
Rather, one would often like to know precisely which noise-noise and noise-disturbance values are obtainable and which are not; that is, to characterise the \emph{noise-noise region}
\begin{align}
	R_{NN}(A,B)=\big\{&\big(N(\M,A),\ N(\M,B)\big) \mid  \notag\\
		& \qquad \M \text{ is a quantum instrument}\big\},
\end{align}
as well as the \emph{noise-disturbance region}
\begin{align}
	R_{ND}(A,B)=\big\{&\big(N(\M,A),\ D(\M,B)\big) \mid  \notag\\
		& \quad \ \ \M \text{ is a quantum instrument}\big\}.
\end{align}

The connection between these two regions and the entropic uncertainty region \begin{equation}\label{eqn:EABdefn}
	E(A,B)=\{\left(H(A|\rho),\, H(B|\rho) \right) \mid \rho \text{ is any density matrix}\}
\end{equation}
will prove fruitful in the search for tighter measurement uncertainty relations.
We therefore find it helpful to distill this connection into the following two propositions, the essence of which can be found implicitly in the arguments contained in the supplemental materials of Refs.~\cite{Buscemi:2014aa} and~\cite{Sulyok:2015aa}.
We provide more direct proofs of these propositions in the Appendix.

Firstly, the noise-noise region can be expressed, as mentioned above, in terms of the entropic uncertainties of the observables in questions.
This proposition follows from the ability to write the noise $N(\M,A)$ in the form of Eq.~\eqref{eqn:NoiseConvEntropyRep}.
\begin{proposition}\label{prop:RNNDef}
	The noise-noise region can be expressed in terms of the entropies of the observables $A$ and $B$ as follows:
	\begin{align}\label{eqn:NNRegionForm}
		R_{NN}(A,B) =& \left\{\sum_m p(m)\big(H(A|\rho_m),\, H(B|\rho_m) \big) \, \Big| \right. \notag\\ 
		& \quad \{p(m),\rho_m\}_m \text{ is a weighted ensemble}\notag\\[-2mm]
		& \quad \left. \text{of states satisfying $\sum_m p(m) \rho_m=\id/d$}\right\}\notag\\
		\subseteq & \conv E(A,B),
	\end{align}
	where $\id$ is the identity operator, $d$ is the Hilbert space dimension, and $\mathrm{conv}\,S$ denotes the convex hull of $S$.
\end{proposition}

Note that this result can also be directly extended to characterise the joint-measurement noise region for three-or-more observables as being included in the convex hull of the corresponding entropic preparation uncertainty region~\cite{Abbott:2016aa}.

Secondly, there is an important relation between the joint-measurement noise and noise-disturbance regions:
the lower boundary of $R_{ND}(A,B)$ always lies on or above the lower boundary of $R_{NN}(A,B)$.
More formally, we have the following proposition relating $R_{NN}(A,B)$ and $R_{ND}(A,B)$.
\begin{proposition}\label{prop:NN-ND-relation}
	For any observables $A,B$ one has 
	\begin{equation}
		R_{ND}(A,B)\subseteq \cl R_{NN}(A,B),
	\end{equation} 
	where $\cl$ denotes the monotone closure (i.e., the closure under increasing either coordinate) up to to the trivial upper bounds $N(\M,A),D(\M,B)\le \log d$.
\end{proposition}
Note that it need not be the case that $R_{ND}(A,B)\subseteq R_{NN}(A,B)$ in general.
For example, in the scenario depicted in Fig.~\ref{fig:NNregion}(b), the point $(N(\M,A),\,N(\M,B))=(1,0)$ is not contained in $R_{NN}(A,B)$, whereas $(N(\M,A),\,\allowbreak D(\M,B))=(1,0)$ is, for qubit measurements, always contained in $R_{ND}(A,B)$ since one can have an instrument that performs the identity transformation and generates a random output.

\subsection{Qubit measurement uncertainty relations}

The relationship between the measurement uncertainty regions and the entropic preparation uncertainty region opens the possibility of providing tighter noise-noise and noise-disturbance uncertainty relations.
Indeed, many of the known state-independent entropic preparation uncertainty relations (e.g., see Refs.~\cite{Vicente:2008oq,Coles:2015ef}) could be used to improve upon Eqs.~\eqref{eqn:MUNN} and~\eqref{eqn:MUND}.
However, such relations that are applicable to arbitrary systems are generally still far from being tight.
For simpler systems such as qubits, on the other hand, much better characterisations are generally possible and of particular interest~\cite{Abdelkhalek:2015cr,Abbott:2016aa}.

In Ref.~\cite{Sulyok:2015aa} the following noise-disturbance relation was proposed for the orthogonal Pauli observables $\sigma_z$ and $\sigma_x$:
\begin{equation}\label{eqn:NDrelnOrthog}
	g\big( N(\M,\sigma_z) \big)^2 +g\big(D(\M,\sigma_x) \big)^2 \le 1,
\end{equation}
where $g$ is the inverse of the function $h$ defined for $x\in[0,1]$ as
\begin{equation}
	h(x)=-\frac{1+x}{2}\log\left(\frac{1+x}{2}\right) -\frac{1-x}{2}\log\left(\frac{1-x}{2}\right).
\end{equation}
Unfortunately, the proof given for this relation was incorrect.
In Section~\ref{sec:NDcharacterisation} we will show that, in fact, this relation does not hold in general, and conjecture a tight bound for the noise-disturbance region.
However, we will also see that the relation does hold in some particular restricted cases of interest, in particular when the measurement is performed by a L\"uders instrument for which the state is simply transformed according to the ``square-root measurement dynamics''.

The approach used to try and prove this relation, given in the supplemental material of~\cite{Sulyok:2015aa}, essentially attempts to show first that Eq.~\eqref{eqn:NDrelnOrthog} characterises the lower boundary of $R_{NN}(\sigma_z,\sigma_x)$, before making use 
of Proposition~\ref{prop:NN-ND-relation} and the fact that that Eq.~\eqref{eqn:NDrelnOrthog} can be saturated to show that it thus also characterises the lower boundary of $ $ $R_{ND}(\sigma_z,\sigma_x)$.

To see that this relation cannot be correct, we first note (a proof is given in the Appendix) that the restriction of $\sum_m p(m)\rho_m=\id/d$ on the weighted ensemble in Proposition~\ref{prop:RNNDef} can be disregarded for the case of qubits, and thus equality is obtained in Eq.~\eqref{eqn:NNRegionForm}.
\begin{proposition}\label{prop:QubitRNNDef}
	For qubits and observables $A$, $B$, the noise-noise region $R_{NN}(A,B)$ is given by
	\begin{align}\label{eqn:qubitNNRegion}
		R_{NN}(A,B) = \conv E(A,B).
	\end{align}
\end{proposition}

Written in this form, it is clear that $R_{NN}(\sigma_z,\sigma_x)$ is a convex set, whereas Eq.~\eqref{eqn:NDrelnOrthog} characterises a (strictly) concave set (see Fig.~\ref{fig:NNregion}(a)) and therefore cannot be the lower boundary of this region, thus undermining the proof given in Ref.~\cite{Sulyok:2015aa}.\footnote{Specifically, the error in the proof lies in the fact that the optimal values of $\theta_m$ ($=0$ or $\pi$) give denominators in Eq.~(6) of the supplemental material of~\cite{Sulyok:2015aa} that are $0$. Subsequent to our identification of this error the authors of~\cite{Sulyok:2015aa} published an erratum~\cite{Sulyok:2016aa} acknowledging it and showing that Eq.~\eqref{eqn:NDrelnOrthog} nevertheless holds in the specific case of dichotomic measure-and-prepare instruments (see Sec.~\ref{sec:NDdichotomic} for further discussion).}

\section{Joint-measurement uncertainty relations for qubits}

Before turning again to qubit noise-disturbance uncertainty relations, we will first make use of Proposition~\ref{prop:QubitRNNDef}, along with recent results on tight preparation uncertainty relations for qubits, to formulate tight noise-noise uncertainty relations for arbitrary qubit observables.
The case of joint-measurement noise for qubits is not only of independent interest, but such a characterisation of the noise-noise region will allow us, by making use of Proposition~\ref{prop:NN-ND-relation}, to start to characterise the noise-disturbance region as well.

\subsection{Arbitrary measurements}

Let $A=\vect{a}\cdot\bm{\sigma}$ and $B=\vect{b}\cdot\bm{\sigma}$ be arbitrary Pauli observables (where $\vect{a},\vect{b}$ are unit vectors on the Bloch sphere and $\bm{\sigma}=(\sigma_x,\sigma_y,\sigma_z)$). 
In a recent article~\cite{Abbott:2016aa}, it was shown that the qubit preparation uncertainty region $E(A,B)$ can be completely characterised by the tight preparation uncertainty relation in terms of standard deviations
\begin{align}\label{eqn:SDPrepReln}
	(\Delta A)^2 +  (\Delta B)^2 \,+\, & 2|\vect{a}\cdot\vect{b}|\sqrt{1-(\Delta A)^2}\sqrt{1-(\Delta B)^2} \notag\\
	& \qquad\qquad\qquad \ge 1 + (\vect{a}\cdot\vect{b})^2,
\end{align}
or its equivalent form in terms of entropies
\begin{align}\label{eqn:entropicPrepReln}
	& g\left( H(A|\rho) \right)^2 + g\left( H(B|\rho) \right)^2 \,\notag\\
	& -\, 2\,|\vect{a}\cdot\vect{b}|\,g\left( H(A|\rho) \right)\,g\left( H(B|\rho) \right)
	 \le 1 - (\vect{a}\cdot\vect{b})^2,
\end{align}
with the function $g$ as defined after Eq.~\eqref{eqn:NDrelnOrthog} above.

Relation~\eqref{eqn:entropicPrepReln}, along with Proposition~\ref{prop:QubitRNNDef}, can thus be used to give the following, tight, joint-measurement uncertainty relation for qubits.

\begin{theorem}\label{thm:NNrelnTight}
	Let $A=\vect{a}\cdot\bm{\sigma}$ and $B=\vect{b}\cdot\bm{\sigma}$ be two Pauli observables, and $\M$ an arbitrary quantum instrument.
	Then the values of $N(\M,A)$ and $N(\M,B)$ are contained in the noise-noise region
	\begin{align}\label{eqn:NNrelntight}
		R_{NN}(A,B)=&\conv\big\{(s,t) \mid g\left( s \right)^2 + g\left( t \right)^2\notag\\
		& - 2|\vect{a}\cdot\vect{b}|\,g\left( s \right)\,g\left( t \right) \le 1 - (\vect{a}\cdot\vect{b})^2\big\}.		
	\end{align}
\end{theorem}
Interestingly, the region $E(A,B)$ is non-convex for $|\vect{a}\cdot\vect{b}|\lesssim 0.391$ and convex for $|\vect{a}\cdot\vect{b}|\gtrsim 0.391$~\cite{Sanchez-Ruiz:1998by,Vicente:2008oq,Abbott:2016aa}.
Thus, for $|\vect{a}\cdot\vect{b}|\gtrsim 0.391$, Eq.~\eqref{eqn:NNrelntight} can be expressed explicitly as the \emph{tight} uncertainty relation
\begin{align}\label{eqn:NNrelntightconvex}
	& g\left( N(\M,A) \right)^2 + g\left( N(\M,B) \right)^2 \notag\\
	& - 2|\vect{a}\cdot\vect{b}|\,g\left( N(\M,A) \right)\,g\left( N(\M,B) \right) \le 1 - (\vect{a}\cdot\vect{b})^2.
\end{align}
For $|\vect{a}\cdot\vect{b}|\lesssim 0.391$ no analytic form for the convex hull of $E(A,B)$ exists in general.
However, for $\vect{a}\cdot\vect{b}=0$, i.e., for orthogonal Pauli measurements such as $\sigma_z$ and $\sigma_x$, this can be given explicitly and we have the simple \emph{tight} relation
\begin{equation}\label{eqn:NNboundlinear}
	N(\M,\sigma_z) + N(\M,\sigma_x) \ge 1,
\end{equation}
which is precisely the bound~\eqref{eqn:MUNN} obtained in~\cite{Buscemi:2014aa}.
The region $R_{NN}(A,B)$ is shown in Figure~\ref{fig:NNregion} for two values of $\vect{a}\cdot\vect{b}$, along with the region $E(A,B)$.

\begin{figure*}[ht]
	\begin{center}
		\begin{tabular}{ccc}
			\includegraphics[width=0.45\textwidth]{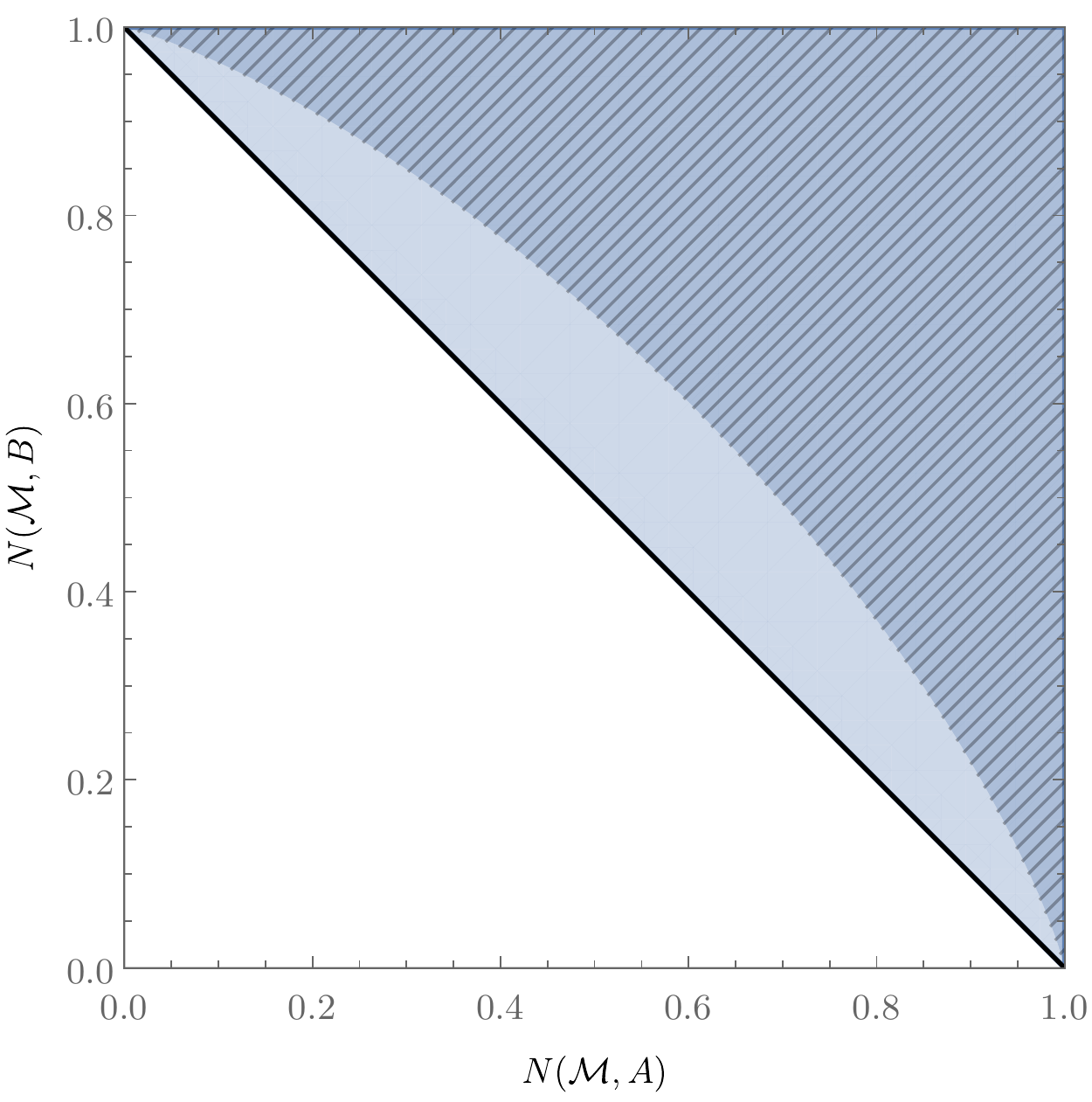}& \qquad\qquad &\includegraphics[width=0.45\textwidth]{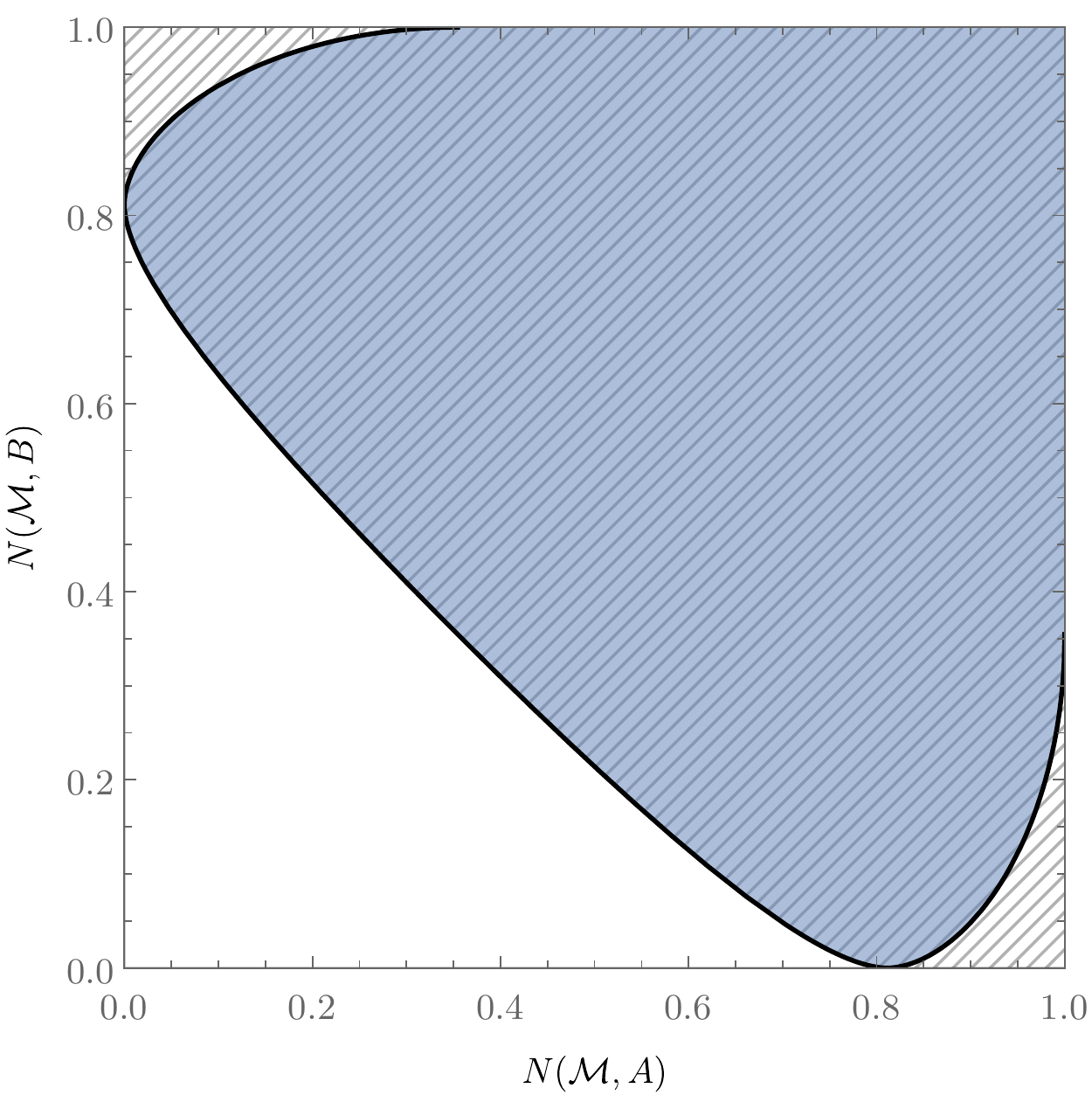}\\
			(a) &  & (b)\\
		\end{tabular}
	\end{center}
	\caption{The shaded area represents the allowable values of $\big(N(\M,A),\, N(\M,B)\big)$ for observables $A=\vect{a}\cdot\bm{\sigma}$ and $B=\vect{b}\cdot\bm{\sigma}$ where (a) $\vect{a}\cdot\vect{b}=0$ and (b) $\vect{a}\cdot\vect{b}=\frac{1}{2}$. The black line represents the uncertainty relation~\eqref{eqn:NNrelntight}, the darker shaded area is the entropic preparation uncertainty region $E(A,B)$ bounded by Eq.~\eqref{eqn:entropicPrepReln} (which reduces to the form of Eq.~\eqref{eqn:NDrelnOrthog} for $\vect{a}\cdot\vect{b}=0$), and the hatched area represents $\cl E(A,B)$. Note that for $\vect{a}\cdot\vect{b}=\frac{1}{2}$ the region $E(A,B)$ is convex and thus $R_{NN}(A,B)=E(A,B)$, whereas $E(A,B)\subsetneq R_{NN}(A,B)$ for the case $\vect{a}\cdot\vect{b}=0$ .}
	\label{fig:NNregion}
\end{figure*}

In order to see that the characterisation of $R_{NN}(A,B)$ given in Eq.~\eqref{eqn:NNrelntight} is indeed tight, one can check that any point $(s,t)\in R_{NN}(A,B)=\conv E(A,B)$ can be obtained by some $\M$.
Let us first consider the case that $(s,t)\in E(A,B)$.
Let $\rho_+=\frac{1}{2}(\id+\vprod{r}{\sigma})$ be a qubit state giving the measurement entropies $(H(A|\rho_+),H(B|\rho_+))=(s,t)$ and $\rho_-=\frac{1}{2}(\id-\vprod{r}{\sigma})$.
Since $H(A|\rho_+)=H(A|\rho_-)$, from Eq.~\eqref{eqn:NoiseConvEntropyRep} we thus have $N(\M,A)=H(A|\rho_+)$, and similarly so for $B$.
Hence, any measurement apparatus $\M$ implementing the POVM $\{\rho_+,\rho_-\}$ has $(N(\M,A)$, $N(\M,B))=(s,t)$, 
as desired.

To show that any point in $R_{NN}(A,B)\setminus E(A,B)$ can also be obtained (which perhaps corresponds to the case of most interest), we need to make use of POVMs with more outcomes.
Since any such point $(s,t)$ is in the convex hull of $E(A,B)$, it can be expressed as a convex combination $q(s_1,t_1) + (1-q)(s_2,t_2)$ of the points $(s_1,t_1),(s_2,t_2)\in E(A,B)$ with $q\in[0,1]$.
Let $\{\rho_{1+},\rho_{1-}\}$ and $\{\rho_{2+},\rho_{2-}\}$ be two POVMs that allow $(s_1,t_1)$ and $(s_2,t_2)$ to be obtained, respectively, as above.
Then an apparatus $\M$ implementing the POVM
\begin{equation}\label{eqn:MsaturatingNN}
	\big\{q\rho_{1+},\ q\rho_{1-},\ (1-q)\rho_{2+},\ (1-q)\rho_{2-}\big\}
\end{equation}
which performs a combination of these two measurements with probabilities $q$ and $(1-q)$, respectively, gives 
\begin{align}
	&(N(\M,A),N(\M,B)) \notag\\ & \qquad = (qu_1 + (1-q)u_2, qv_1 + (1-q)v_2)
	 =(s,t),
\end{align}
thus allowing any point in $R_{NN}(A,B)$ to be realised (in particular, those on its boundary).

The above construction for obtaining points contained in $R_{NN}(A,B)\setminus E(A,B)$ uses four-outcome POVMs, which raises the question of whether the same set of noise-noise values can be obtained if the number of measurement outcomes is restricted.
Below, we will show that dichotomic measurements can only give noise-noise values contained in $\cl(E(A,B))$, 
a realisation that further motivates our investigation, in Section~\ref{sec:NDcharacterisation}, of the form of the noise-disturbance region and whether Eq.~\eqref{eqn:NDrelnOrthog} may hold, at least under certain conditions.
Numerical simulations with random POVMs appear to show that the region of noise-noise values obtainable with three-outcome POVMs lies in between those obtainable with two- and four-outcome POVMs, and thus that four-outcome measurements are indeed required to saturate the noise-noise tradeoff when $E(A,B)$ is not convex, but we leave further clarification of this point to future work.

\subsection{Dichotomic measurements}\label{sec:NNdichotomic}

Let us denote the restriction of the noise-noise region to two-outcome measurements $R^*_{NN}(A,B)$.
In order to find the lower boundary of this region -- and thus tight uncertainty relations on the joint-measurement noise for such measurements -- it is first important to note that the reduction from Proposition~\ref{prop:RNNDef} to Proposition~\ref{prop:QubitRNNDef} for qubits does not hold if the number of outcomes is fixed (cf.\ the proof in the Appendix).
Thus, for dichotomic measurements we must make use of Eq.~\eqref{eqn:NNRegionForm}, with the restriction $\sum_m p(m)\rho_m=\frac{1}{2}\id$.

If we label the measurement outcomes $\pm$ and write $\rho_\pm=\frac{1}{2}(\id+k_\pm\vect{\hat{r}}_\pm\cdot\vect{\sigma})$, for some unit vectors $\vect{\hat{r}}_\pm$ and $k_\pm\in[0,1]$, then this normalisation condition ensures that $\vect{\hat{r}}_+=-\vect{\hat{r}}_-$.
From Eq.~\eqref{eqn:NoiseConvEntropyRep} we see that the noise $N(\M,A)$ then satisfies $N(\M,A)=\sum_m p(m)H(A|\rho_m)\ge H(A|\hat{\rho}_+)=H(A|\hat{\rho}_-)$, where $\hat{\rho}_\pm=\frac{1}{2}(\id + \vect{\hat{r}}_\pm\cdot\vect{\sigma})$, and similarly for $N(\M,B)$.
The region $R^*_{NN}(A,B)$ must therefore be contained in the monotone closure of the entropic uncertainty region $E(A,B)$:
\begin{align}\label{eqn:NNdichotomicRegn}
	R^*_{NN}(A,B) &\subseteq \cl E(A,B).
\end{align}
Combined with the fact that any point in $E(A,B)$, and in particular those on its boundary, can be reached by noise-noise values,%
\footnote{One can indeed easily see that any point $(H(A|\rho),H(B|\rho))$ in $E(A,B)$, for some state $\rho$, can be reached by the values $(N(\M,A),N(\M,B))$ for a dichotomic instrument $\M$ implementing the POVM $\{\rho,\id-\rho\}$.} 
this shows that the lower boundaries of $R^*_{NN}(A,B)$ and $E(A,B)$ coincide.%
\footnote{In fact, one can show the stronger claim that $R^*_{NN}(A,B)=E(A,B)$. To see this, note first that $R^*_{NN}(A,B)\subseteq R_{NN}(A,B)=\conv E(A,B)$. 
From the characterisation of $E(A,B)$ one can further show that $\conv E(A,B)\cap \cl E(A,B)=E(A,B)$ (this can readily be seen to be the case visually, although the formal proof is a little tedious) and one thus has $R^*_{NN}(A,B)\subseteq E(A,B)$.
From the previous footnote, $E(A,B)\subseteq R_{NN}^*(A,B)$, which concludes the proof.} 

For dichotomic measurements, the uncertainty relation~\eqref{eqn:NNrelntightconvex} thus holds and is tight \emph{for all} $\vect{a},\vect{b}$.
For the orthogonal Pauli observables $\sigma_z$ and $\sigma_x$, this reduces to the simple, tight relation (analogous to Eq.~\eqref{eqn:NDrelnOrthog})
\begin{equation}\label{eqn:NNboundDichotomic}
	g(N(\M,\sigma_z))^2 + g(N(\M,\sigma_x))^2 \le 1.
\end{equation}

\section{Noise-disturbance uncertainty relations for orthogonal qubit measurements}\label{sec:NDcharacterisation}

The error in the proof of Eq.~\eqref{eqn:NDrelnOrthog} given in Ref.~\cite{Sulyok:2015aa}, along with the differences between the region defined by this relation and the noise-noise region $R_{NN}(\sigma_z,\sigma_x)$ bounded by Eq.~\eqref{eqn:NNboundlinear}, raises the question of whether the noise-disturbance tradeoff can be decreased below Eq.~\eqref{eqn:NDrelnOrthog}.
In this section we first show that this bound can indeed by violated, before looking at characterising the noise-disturbance region, as well as its form under certain natural restrictions.

\subsection{Violating Eq.~\eqref{eqn:NDrelnOrthog}}\label{sec:counterexample}

Consider the three-outcome measurement $\M^\theta$ with the associated POVM $M^\theta=\{M^\theta_{-1},M^\theta_{0},M^\theta_1\}$ for $\theta\in[0,\pi/2]$, where $M^\theta_m=p_m(\id+\vect{n}_m\cdot\vect{\sigma})$ and $\vect{n}_m=((-1)^m \cos (m\theta),0$, $\sin (m\theta))$, $p_{0}=\frac{\cos\theta}{1+\cos\theta}$ and $p_{-1}=p_1=\frac{1}{2(1+\cos\theta)}$.
One can readily verify that this is a valid POVM.
The probability of obtaining outcome $m$ when measuring a state $\rho$ is thus $\Tr[M_m \rho]$, and we consider the case that, following the measurement, the system is in the pure state $\ket{n_m}$ with Bloch vector $\vect{n}_m$.

From Eq.~\eqref{eqn:NoiseConvEntropyRep} we can calculate the noise on $\sigma_z$ to be
\begin{equation}
	N(\M^\theta,\sigma_z)=\frac{\cos\theta+h(\sin\theta)}{1+\cos\theta}.
\end{equation}

In order to determine an upper bound on the disturbance $D(\M^\theta,\sigma_x)$, let us consider the correction $\mathcal{E}=\{\mathcal{E}_m\}_m$ that leaves the state unchanged on outcome 0, and maps $\vect{n}_{-1}$ and $\vect{n}_1$ onto the negative $x$-axis.
One may implement this with unitary transformations, or, more simply, require that $\mathcal{E}_0(\rho)=\rho$ and $\mathcal{E}_{-1}(\rho)=\mathcal{E}_{1}(\rho)=\frac{1}{2}(\id-\sigma_x)$ for all $\rho$.
From Eqs.~\eqref{eqn:jointDistBB} and~\eqref{eqn:jointDistBBcond} one can then calculate the joint distribution $p(b',b)$ and thus the upper bound on the disturbance as
\begin{equation}
	D_\E(\M^\theta,\sigma_x)=\frac{h(\cos\theta)}{1+\cos\theta}.
\end{equation}

This measurement-correction pair violates Eq.~\eqref{eqn:NDrelnOrthog} for all $\theta\in(0,\pi/2)$.
Taking, for example, $\theta=\frac{\pi}{3}$ we find $g(N(\M^\theta,\sigma_z))^2 \,+\, g(D_\E(\M^\theta,\sigma_x))^2\approx 1.1 > 1$;
thus since $D(\M^\theta,\sigma_x)\le D_\E(\M^\theta,\sigma_x)$ and $g$ is a decreasing function, Eq.~\eqref{eqn:NDrelnOrthog} is clearly violated.

The curve given parametrically by
\begin{align}
(N(\M,\sigma_z),D(\M,\sigma_x))=\left(\tfrac{\cos\theta+h(\sin\theta)}{1+\cos\theta},
 \tfrac{h(\cos\theta)}{1+\cos\theta}\right)
\end{align}
for $0\le\theta\le\frac{\pi}{2}$ is thus an upper bound for the lower boundary of $R_{ND}(\sigma_z,\sigma_x)$.
This bound, which is shown in Fig.~\ref{fig:NDbound}, is asymmetric around the line $N(\M,\sigma_z)=D(\M,\sigma_x)$, in contrast to the tight bounds for the joint-measurement relations shown in Fig.~\ref{fig:NNregion}.

\begin{figure}[t]
	\begin{center}
			\includegraphics[width=\columnwidth]{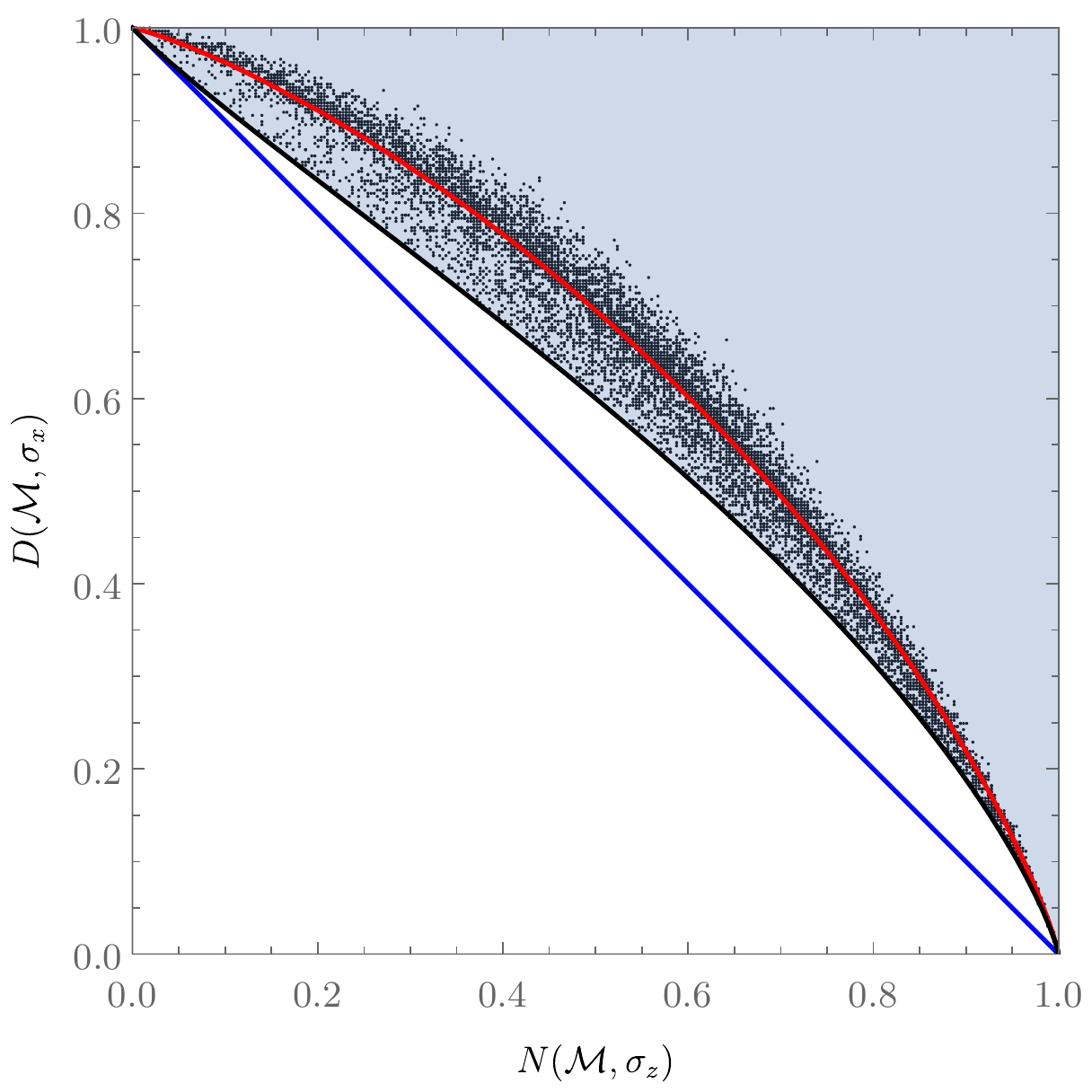}
	\end{center}
	\caption{The shaded area represents the conjectured noise-disturbance region \eqref{eqn:NDconjecturebound}, the lower boundary of which is attained by $\M^\theta$ for $\theta\in[0,\pi/2]$ (black line), as described in the text; the blue line is the noise-noise bound \eqref{eqn:NNboundlinear}; and the red line is the (in general incorrect) bound \eqref{eqn:NDrelnOrthog} from Ref.~\cite{Sulyok:2015aa}. The noise-disturbance points plotted correspond to ten thousand random three- and four-outcome instruments with POVMs restricted to the $xz$-plane and unitary corrections obtained by a heuristic optimisation procedure (see Appendix).}
	\label{fig:NDbound}
\end{figure}

\subsection{Characterising the noise-disturbance region}\label{sec:characterisingRND}

With this proof that Eq.~\eqref{eqn:NDrelnOrthog} can be violated, the problem of characterising precisely the lower boundary of $R_{ND}(\sigma_z,\sigma_x)$ (and, more generally, $R_{ND}(A,B)$) is opened up once more.
While Proposition~\ref{prop:NN-ND-relation} places a lower bound on this tradeoff, there is no immediately obvious way to saturate the boundary of $\cl(R_{NN}(A,B))$ with noise-disturbance values, and the search for a tight characterisation of $R_{ND}(A,B)$ thus requires a careful analytic analysis of the noise-disturbance tradeoff, a problem significantly more complicated than in the noise-noise case.

Perhaps the most immediate problem in attempting such an analysis is the fact that one must minimise over all possible corrections in order to calculate the disturbance for a given measurement.
However, by noting that it is always possible to incorporate the optimal correction $\mathcal{E}$ into the transformation performed by an instrument to yield another valid instrument, we see that, for any instrument $\M$, there is another instrument $\M'$ such that $N(\M,A)=N(\M',A)$ and $D(\M,B)=D_\mathcal{I}(\M',B)$, where $\mathcal{I}$ is the identity correction, and thus represents the equivalent case where no correction is applied.
If we define the noise-disturbance region of this restricted no-correction scenario as
\begin{align}
	R_{ND_\mathcal{I}}(A,B)&=\big\{ \big( N(\M,A),\, D_\mathcal{I}(\M,B)\big) \mid \notag\\
	& \qquad \text{$\M$ is a quantum instrument} \big\},
\end{align}
we therefore have $R_{ND}(A,B)\subseteq R_{ND_\mathcal{I}}(A,B)$ and hence $\cl R_{ND}(A,B)\subseteq \cl R_{ND_\mathcal{I}}(A,B)$ also.
Moreover, since $D_\mathcal{I}(\M,B)\ge D(\M,B)$ we also see that $R_{ND_\mathcal{I}}(A,B) \subseteq \cl R_{ND}(A,B)$.
Finally, by noting that $R_{ND_\mathcal{I}}(A,B)=\cl R_{ND_\mathcal{I}}(A,B)$ (since, when no correction is applied, one can always add noise to an instrument without increasing $D_\mathcal{I}(\M,B)$, or conversely mix the outgoing state with the completely mixed state to increase $D_\mathcal{I}(\M,\sigma_x)$ without increasing the noise) and putting these steps together we find that
\begin{equation}\label{eqn:NDnoCorrCharac}
	\cl R_{ND}(A,B)=R_{ND_\mathcal{I}}(A,B).
\end{equation}  
The lower boundaries of these region thus coincide, and we can restrict ourselves to considering the uncorrected disturbance $D_\mathcal{I}(\M,B) = H(\mathbb{B}|\mathbb{B}_{\M,\mathcal{I}}')$ in order to characterise the noise-disturbance tradeoff.\footnote{However, in general $\cl(R_{ND}(A,B)) \neq R_{ND}(A,B)$, and thus $R_{ND}(A,B)\neq R_{ND_\mathcal{I}}(A,B)$ (although the two sets may coincide in some particular cases, like for orthogonal $A$ and $B$). To see this, note for instance that if $A=B$ one can have $(N(\M,A),D(\M,B))=(0,\delta)$ only if $\delta=0$, while $(N(\M,A),$ $D_\mathcal{I}(\M,B))=(0,\delta)$ can be obtained for any $\delta\in[0,1]$.}

The problem is still rather complicated since, for every POVM $M$, one must consider all possible transformations that can be performed by the instrument, each giving rise to a different instrument $\M=\{\M_m\}_m$.
If the Kraus operators corresponding to each $\M_m$ are $\{K_{m,i}\}_{i=1}^{n_m}$ (so that $\M_m(\rho)=\sum_{i=1}^{n_m}K_{m,i}\rho K_{m,i}^\dagg$), then the problem can be further simplified by noting that we can write $K_{m,i}=U_{m,i}\sqrt{M_{m,i}}$ where $U_{m,i}$ is a unitary and $\sqrt{M_{m,i}}$ is the (unique) positive semidefinite Hermitian root of $M_{m,i}=K_{m,i}^\dagg K_{m,i}$.
We can then consider another instrument $\M'$ with $\sum_{m}n_m$ outcomes such that $\M'_{m,i}(\rho)=K_{m,i}\rho K_{m,i}^\dagg$; i.e., each outcome is associated with a POVM element $M'_{m,i} = M_{m,i}$ and the corresponding transformation has a single Kraus operator $K_{m,i}$.
Such an instrument is said to be \emph{purity preserving}.
Note that this instrument can be interpreted as a L\"{u}ders instrument with an additional unitary correction $U_{m,i}$ applied depending on the measurement outcome (cf.\ Sec.~\ref{sec:NDLuders}).
If we define the restriction of the noise-disturbance region to purity-preserving instruments as
\begin{align}
	&R_{ND_\mathcal{I}}^{PP}(A,B)=\big\{ \big( N(\M,A),\, D_\mathcal{I}(\M,B)\big) \mid \notag\\
	& \quad \text{$\M$ is a purity-preserving quantum instrument} \big\},
\end{align}
then we clearly have $R_{ND_\mathcal{I}}^{PP}(A,B) \subseteq R_{ND_\mathcal{I}}(A,B)$, since these are just a subset of all instruments, and thus also $\cl R_{ND_\mathcal{I}}^{PP}(A,B) \subseteq R_{ND_\mathcal{I}}(A,B)$.
Conversely, we see that $D_\mathcal{I}(\M,B)=D_\mathcal{I}(\M',B)$ since we constructed $\M'$ such that $\M'(\oprod{\pm b}{\pm b})=\M(\oprod{\pm b}{\pm b})$, and, moreover, the fact that the POVM $\{M_m\}_m$ is simply a coarse graining of $\{M'_{m,i}\}_{m,i}$ implies that, from the definition of noise and a simple application of the classical data-processing inequality, $N(\M,A)\ge N(\M',A)$.
It thus follows that $R_{ND_\mathcal{I}}(A,B) \subseteq \cl R_{ND_\mathcal{I}}^{PP}(A,B)$ and hence these two sets are the same.
Combining this with Eq.~\eqref{eqn:NDnoCorrCharac} we have
\begin{equation}
	\cl R_{ND}(A,B)=\cl R_{ND_\mathcal{I}}^{PP}(A,B),
\end{equation}
and thus the lower boundaries of these regions coincide.

For purity-preserving qubit instruments the calculation of $D_\mathcal{I}(\M,B)$ is somewhat simplified, and in the Appendix we give an analytic formula for it in terms of the POVM elements $M_m$ and the unitaries $U_m$.
The noise-disturbance tradeoff can thus be characterised by considering all POVMs $M=\{M_m\}_m$ and the unitaries $\{U_m\}_m$ that minimise $D_\mathcal{I}(\M,B)$ for each such POVM when $\M$ performs the transformation above.
Unfortunately, there does not appear to be any simple way to analytically determine the optimal such unitaries, and as a result we were not able to prove a tight bound for $R_{ND}(A,B)$, even for the case of orthogonal $A$ and $B$.

From now on we will pursue just this case of orthogonal Pauli observables, fixing $A=\sigma_z$ and $B=\sigma_x$ and leaving the more general case to future work.
Despite our inability to analytically characterise $R_{ND}(\sigma_z,\sigma_x)$, it is possible to study its form via numerical simulations by testing large numbers of randomly generated quantum instruments.
Na\"{i}vely generating such instruments generally results in most instruments being far from the boundary of the region.
However, making use of the above simplifications it is possible to do much better by randomly generating POVMs (rather than instruments) and using numerical approaches to finding the (or close to the) optimal set of unitaries for each POVM.
Some care nonetheless still needs to be made in choosing the distribution from which to draw POVMs from, and further details of our approach are given in the Appendix.

We performed extensive such numerical simulations for measurements with 2 to 6 outcomes, and the results of some of these (for 3 and 4 outcome measurements) are shown in Fig.~\ref{fig:NDbound}.
Our results suggest that the bound obtained from the counter-example in the previous section is in fact tight, as not a single instrument violating it was found.
We thus formulate the following conjecture. 
\begin{conjecture}
	Let $\M$ be an arbitrary quantum instrument for qubits.
	Then the values of $N(\M,\sigma_z)$ and $ $ $D(\M,\sigma_x)$ are contained in the noise-disturbance region
	\begin{align}\label{eqn:NDconjecturebound}
		R_{ND}(\sigma_z,\sigma_x) = &  \cl\left\{\left(\tfrac{\cos\theta+h(\sin\theta)}{1+\cos\theta},\tfrac{h(\cos\theta)}{1+\cos\theta}\right) \big|\  0\le \theta \le \pi/2\right\}.
	\end{align}
\end{conjecture}
This conjecture, if correct, would be surprising since it would indicate that, in stark contrast to the case of joint-measurement noise, three-outcome measurements are sufficient to completely saturate the noise-disturbance bound, and could thus be said to be optimal in this respect.

\subsection{Dichotomic measurements}\label{sec:NDdichotomic}

While we found that it was possible to saturate the conjectured bound for $R_{ND}(\sigma_z,\sigma_x)$ with measurements with three or more outcomes, there seemed no apparent way to do so with dichotomic measurements, and thus it seems that (at least) three outcome measurements are not only sufficient but also necessary to saturate the noise-disturbance tradeoff.
Given the fact that the noise-noise region is bounded by Eq.~\eqref{eqn:NNboundDichotomic} along with the relation given in Proposition~\ref{prop:NN-ND-relation}, one may be tempted to think that the restriction of the noise-disturbance region to dichotomic measurements, which we denote $R^*_{ND}(A,B)$, must satisfy $R_{ND}^*(A,B)\subseteq \cl R_{NN}^*(A,B) \subseteq \cl E(A,B)$ and thus that Eq.~\eqref{eqn:NDrelnOrthog} holds for dichotomic measurements.
Indeed, in a recent erratum~\cite{Sulyok:2016aa} acknowledging the error in their proof~\cite{Sulyok:2015aa} of Eq.~\eqref{eqn:NDrelnOrthog}, the authors prove that this is the case for the subset of dichotomic measurements that are of the `measure-and-prepare' form, which includes the measurements performed in their experimental tests of Eq.~\eqref{eqn:NDrelnOrthog}.

However, the argument used to prove Proposition~\ref{prop:NN-ND-relation} does not hold if the number of outcomes is fixed (cf. the discussion in the Appendix), so such reasoning would be premature.
A more careful analysis showed that it is in fact possible to violate Eq.~\eqref{eqn:NDrelnOrthog} with carefully chosen dichotomic measurements and corrections.
Specifically, consider the POVM $M=\{M_\pm\}_\pm$ with $M_+=\frac{1}{4}(\id + \frac{1}{\sqrt{2}}(\sigma_x+\sigma_z))$ and $\M_-=\id-M_+$, and the associated instrument $\M$ implementing the transformation $\M_\pm(\rho)=\sqrt{M_\pm}\rho \sqrt{M_\pm}$, and consider a correction $\E_+(\rho)=\oprod{x}{x}$ applied on outcome `$+$' (and no correction for the other outcome).
From Eqs.~\eqref{eqn:NoiseConvEntropyRep}, \eqref{eqn:jointDistBB} and~\eqref{eqn:jointDistBBcond} one can calculate that $N(\M,\sigma_z)\approx 0.870$ and $D_\E(\M,\sigma_x)\approx 0.255$ which gives $g(N(\M,\sigma_z))^2+g(D_\E(\M,\sigma_x))^2\approx 1.011 > 1$.
By considering different instruments and optimising over corrections we were able to do marginally better than this, although the instruments and corrections doing so are not particularly informative; the best violation of Eq.~\eqref{eqn:NDrelnOrthog} we found numerically gave $g(N(\M,\sigma_z))^2+g(D_\E(\M,\sigma_x))^2\approx 1.024$. 

Although such a violation is rather small it is still perhaps surprising, given the results for the noise-noise case and for measure-and-prepare instruments~\cite{Sulyok:2016aa}, that $R^*_{ND}(\sigma_z,\sigma_x)\neq R^*_{NN}(\sigma_z,\sigma_x)$.
Figure~\ref{fig:NDdichotomic} shows the results of numerical simulations with dichotomic measurements in relation to the bounds~\eqref{eqn:NDrelnOrthog} and~\eqref{eqn:NDconjecturebound}.
One can see that the lower boundary of $R^*_{ND}(\sigma_z,\sigma_x)$ appears to be only slightly below that of $R^*_{NN}(\sigma_z,\sigma_x)$.

\begin{figure}[t]
	\begin{center}
			\includegraphics[width=\columnwidth]{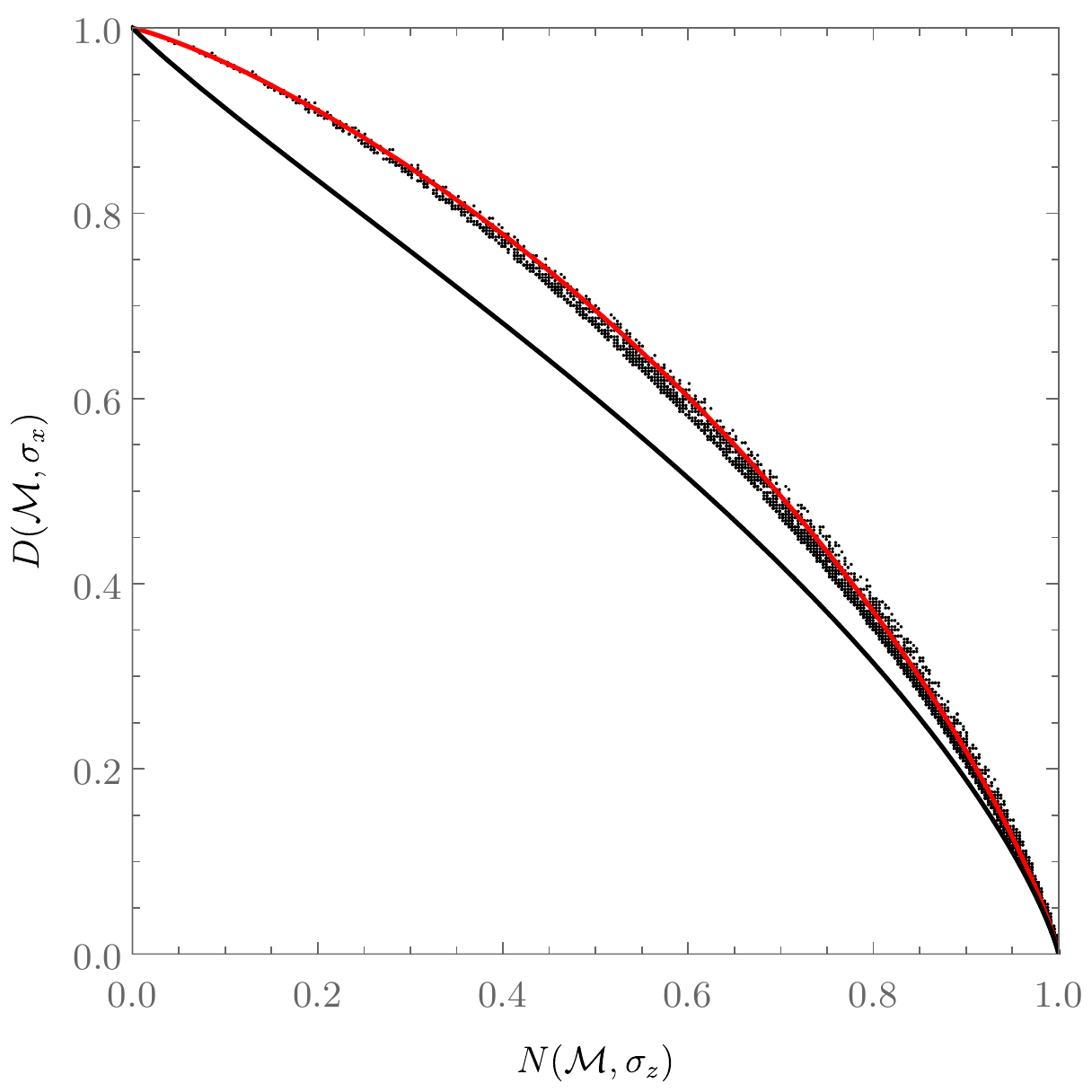}
	\end{center}
	\caption{The points plotted correspond to the values of $N(\M,\sigma_z)$ and $D_\E(\M,\sigma_x)$ for ten thousand random dichotomic instruments with POVMs restricted to the $xz$-plane and numerically optimised unitary corrections applied (see Appendix). The curves correspond to the violated bound \eqref{eqn:NDrelnOrthog} (red line) and the conjectured lower bound of $R_{ND}(\sigma_z,\sigma_x)$ (black line).}
	\label{fig:NDdichotomic}
\end{figure}

\subsection{Noise-disturbance relations for L\"{u}ders instruments}\label{sec:NDLuders}

Although Eq.~\eqref{eqn:NDrelnOrthog} does not hold in general, our simulations showed that it required carefully chosen post-measurement corrections in order to violate it.
In this section we go further and show that it is in fact valid for an interesting class of measurements, in which $\M$ is a ``L\"{u}ders instrument''~\cite{Luders:1951aa} that updates the state according to the so-called ``square-root dynamics''~\cite{Barnum:2000aa}, and no further correction is applied (i.e., when one considers $D_\mathcal{I}(\M,\sigma_x)$ instead of $D(\M,\sigma_x)$).\footnote{Note that the measurements performed by Ref.~\cite{Sulyok:2015aa} saturating Eq.~\eqref{eqn:NDrelnOrthog} were not implemented by L\"{u}ders instruments, as a non-trivial correction was used.}
A measurement instrument $\M$ with associated POVM $\{M_m\}_m$ is a L\"{u}ders instrument if the state is updated according to $\M_m(\rho)=\sqrt{M_m}\rho\sqrt{M_m}$. 
Such measurements can be seen as a generalisation of standard projective measurements~\cite{Luders:1951aa}, and correspond to many realistic experimental situations.

Let $M=\{M_m\}_m$ be an arbitrary qubit POVM as before.
Then we can write each $M_m$ as
\begin{equation}\label{eqn:genPOVMform}
	M_m = p_m(\id + k_m\vect{n}_m\cdot\vect{\sigma}),
\end{equation}
where $|\vect{n}_m|=1$, $p_m\ge 0$ and $|k_m|\le 1$.
The normalisation of $M$, i.e.\ $\sum_m M_m = \id$, is then expressed by the conditions $\sum_m p_m = 1$ and $\sum_m p_m k_m \vect{n}_m = \vect{0}$.

Using this representation we find that the noise $N(\M,\sigma_z)$ for any instrument $\M$ realising the POVM $M$ can be expressed as 
\begin{equation}\label{eqn:NParameterised}
	N(\M,\sigma_z)=\sum_m p_m h\left(|k_m\vect{n}_m\cdot\vect{z}| \right).
\end{equation}

In order to calculate $D_\mathcal{I}(\M,\sigma_x)$ we must first calculate the average post-measurement state
\begin{equation}
	\rho_+ = \M(\oprod{x}{x})=\frac{1}{2}(\id+\vect{r}_+\cdot\vect{\sigma}),
\end{equation}
as well as the similarly defined $\rho_- = \frac{1}{2}(\id+\vect{r}_-\cdot\vect{\sigma})$ for the input $\ket{-x}$.
For a L\"{u}ders instrument, $\M(\oprod{x}{x})=\sum_m\sqrt{M_m}\oprod{x}{x}\sqrt{M_m}$ and we find that
\begin{align}\label{eqn:rvectorFull}
	\vect{r}_\pm = \pm \sum_m p_m  &\left( (\vect{n}_m\cdot\vect{x})\vect{n}_m \vphantom{\sqrt{1-k_m^2}}\right.\notag\\[-3mm]
	&\ \left. + \sqrt{1-k_m^2}(\vect{x}-(\vect{n}_m\cdot\vect{x})\vect{n}_m) \right)
\end{align}
and
\begin{equation}\label{eqn:DParameterised}
	D_\mathcal{I}(\M,\sigma_x)=h(|\vect{r}_+\cdot\vect{x}|).
\end{equation}

We will make use of the following fact, which can easily be verified, to show that a L\"{u}ders instrument, for which the restricted definition of disturbance $D_\mathcal{I}(\M,\sigma_x)$ is employed, must obey Eq.~\eqref{eqn:NDrelnOrthog}.
\begin{fact}\label{convexityLemma}
	The function $f(x)=h(\sqrt{1-x^2})$ is convex on $[0,1]$.
\end{fact}

\begin{theorem}\label{thm:sqrtdynamics}
	Let $\M$ be a L\"{u}ders instrument for qubits.
	Then the following tight relation holds:
	\begin{equation}\label{eqn:NDsqrtdynamics}
		g(N(\M,\sigma_z))^2 + g(D_\mathcal{I}(\M,\sigma_x))^2 \le 1.
	\end{equation}
\end{theorem}
\begin{proof}
	Let us write the $\vect{r}_+$ above as $\vect{r}_+=\sum_m p_m \vect{r}_m$, where 
	\begin{equation}
		\vect{r}_m = (\vect{n}_m\cdot\vect{x})\vect{n}_m + \sqrt{1-k_m^2}\big(\vect{x}-(\vect{n}_m\cdot\vect{x})\vect{n}_m\big).
	\end{equation}
	Let $u_x=\sum_m p_m|\vect{r}_m|$ and define the vector $\vect{u}=u_x\,\vect{x}+\sqrt{1-u_x^2}\,\vect{z}$.
	Since $\vect{n}_m$ and $(\vect{x}-(\vect{n}_m\cdot\vect{x})\vect{n}_m)$ are orthogonal and $|\vect{x}-(\vect{n}_m\cdot\vect{x})\vect{n}_m|^2=1-(\vect{n}_m\cdot\vect{x})^2$ we have $1-|\vect{r}_m|^2 = k_m^2(1-(\vect{n}_m\cdot\vect{x})^2)$.
	Using Eq.~\eqref{eqn:NParameterised} along with the fact that $h$ is decreasing and $|\vect{n}_m\cdot\vect{z}|\le\sqrt{1-|\vect{n}_m\cdot\vect{x}|^2}$, we have
	\begin{align}
		N(\M,\sigma_z) &= \sum_m p_m h\big(|k_m(\vect{n}_m\cdot\vect{z})|\big)\notag\\ 
		& \ge \sum_m p_m h\left(|k_m|\sqrt{1-(\vect{n}_m\cdot\vect{x})^2}\right)\notag\\ 
		& = \sum_m p_m h\left(\sqrt{1-|\vect{r}_m|^2}\right)\notag\\
		&\ge h\left(\sqrt{1-(\textstyle\sum_m p_m|\vect{r}_m|)^2}\right)\notag\\ 
		& = h\left(\sqrt{1-u_x^2}\right) = h(\vprod{u}{z}) = H(\sigma_z|\rho_u),
	\end{align}
	with $\rho_u=\frac{1}{2}(\id+\vect{u}\cdot\vect{\sigma})$, and where we have used Fact~\ref{convexityLemma} to give the second inequality.
	
	Calculating the disturbance for L\"{u}ders instruments, i.e.\ $D_\mathcal{I}(\M,\sigma_x)=H(\mathbb{X}|\mathbb{X}_{\M,\mathcal{I}}')$, we have
	\begin{align}
		D_\mathcal{I}(\M,\sigma_x)&=h\big(|\textstyle\sum_m p_m \vect{r}_m\cdot\vect{x}|\big)\notag\\
		& \ge h\big(\textstyle\sum_m p_m|\vect{r}_m|\big) = h(u_x) = h(\vprod{u}{x})\notag\\
		&=H(\sigma_x|\rho_u).\label{eqn:distBoundLudersInstr}
	\end{align}
	
	We thus see that the noise and disturbance for L\"{u}ders instruments which implement the square-root dynamics can be both bounded below by the entropy of $\sigma_z$ and $\sigma_x$, respectively, for a common state with Bloch vector $\vect{u}$.
	We hence have $(N(\M,\sigma_z),\,\allowbreak D_\mathcal{I}(\M,\sigma_x))\allowbreak\in \cl E(\sigma_z,\sigma_x)$, and the proof of~\eqref{eqn:NDsqrtdynamics} is completed by recalling that the desired relation corresponds precisely to the lower boundary of $E(\sigma_z,\sigma_x)$.
	
Finally, to see that the relation is tight, consider any values $(s,t)$ satisfying $g(s)^2 + g(t)^2 \le 1$. 
One can then check that these can, for instance, be reached by the noise-disturbance values obtained for the (dichotomic) L\"{u}ders instrument $\M$ with POVM elements $M_\pm=\frac{1}{2}\big(\id\pm(g(s) \,\sigma_z + \sqrt{1{-}g(s)^2{-}g(t)^2}\,\sigma_y)\big)$ (recall that, for L\"{u}ders instruments, the POVM elements uniquely determine the instrument): using Eqs.~\eqref{eqn:NParameterised}, \eqref{eqn:rvectorFull} and~\eqref{eqn:DParameterised}, one indeed finds $N(\M,\sigma_z)=s$ and $D_\mathcal{I}(\M,\sigma_x)=t$.	
\end{proof}

The validity of Eq.~\eqref{eqn:NDrelnOrthog} -- or rather, Eq.~\eqref{eqn:NDsqrtdynamics} -- for measurements performed by L\"{u}ders instruments is 
particularly noteworthy in that it shows that this interesting class of measurements is not optimal.
This is in contrast to results showing such measurements to be optimal in other related scenarios: Ref.~\cite{Barnum:2000aa} found them to be optimal with respect to different measures of information gain and disturbance, while Ref.~\cite{Carmeli:2012aa} showed that they implement minimally unsharp sequential joint measurements.
In order to perform an optimal measurement that saturates the noise-disturbance tradeoff bound, one thus needs to consider non-trivial corrections,%
\footnote{One can strengthen Theorem~\ref{thm:sqrtdynamics} a little to show that it holds if a single unitary correction 
is applied irrespective of the measurement outcome (see the Appendix for a proof).} 
as in the counter-example of Section~\ref{sec:counterexample}, or, equivalently, measurements transforming the system according to more complicated dynamics.

\section{Conclusions and future research}

In this paper we have made use of a recently introduced information-theoretic approach to quantifying both the inherent noise in quantum measurements and the disturbance induced by measurements with respect to a subsequent ideal measurement in order to study, in detail, the noise-noise and noise-disturbance tradeoffs in qubit measurements.

Using recently published tight entropic preparation uncertainty relations for arbitrary qubit observables, we completely characterised the degree to which two incompatible Pauli observables can be jointly measured.
Specifically, we showed that the allowable noise-noise region is precisely the convex hull of the corresponding preparation uncertainty region.
These results could readily be extended to more than two observables to give joint-measurement uncertainty relations for three (or more) Pauli observables using the analogous results for entropic preparation uncertainty relations~\cite{Abbott:2016aa}.

We then discussed a recently proposed noise-disturbance uncertainty relation for orthogonal qubit measurements.
We showed that the proof given for this relation in Ref.~\cite{Sulyok:2015aa} was incorrect and provided counter-examples showing that it can be violated even by dichotomic measurements.
We provided a class of three-outcome measurements that we conjectured saturates the optimal noise-disturbance bound, and provided numerical evidence to back this up.
Interestingly, this characterisation of the set of allowable noise-disturbance values only requires three-outcome measurements, in contrast to the case of joint measurement, where measurements with four outcomes seem to be necessary.

Finally, we showed that an important class of measurements -- those performed by a L\"{u}ders instrument -- satisfies the more restrictive noise-disturbance relation of Ref.~\cite{Sulyok:2015aa}, and therefore cannot obtain the optimal qubit noise-disturbance tradeoff.
This broadens the class of measurements known to satisfy this relation well beyond the case of dichotomic measure-and-prepare instruments shown in~\cite{Sulyok:2016aa}, and thus emphasises that, in order to perform optimal measurements with respect to this tradeoff, one must utilise measurements with non-trivial post-measurement corrections to the state.

It remains an open problem to prove whether or not our conjectured noise-disturbance bound~\eqref{eqn:NDconjecturebound} is indeed correct, and it is similarly unknown whether this bound and the noise-noise bound can be simultaneously saturated by a single measurement.
It would also be interesting to compare these results to those known for more traditional root-mean-square error approaches~\cite{Ozawa:2003fh,Busch:2014ts}.
Furthermore, our results on the noise-disturbance tradeoff apply only to orthogonal Pauli measurements, and their generalisation to non-orthogonal measurements and higher-dimensional systems is left to future work.

\begin{acknowledgments}
We thank Michael J. W. Hall for several discussions and references related to this research.
AA and CB acknowledge financial support from the ``Retour Post-Doctorants'' program (ANR-13-PDOC-0026) of the French National Research Agency; CB also acknowledges the support of a Marie Curie International Incoming Fellowship (PIIF-GA-2013-623456) from the European Commission.
\end{acknowledgments}

\appendix
\renewcommand{\thesection}{}
\renewcommand{\theequation}{A\arabic{equation}}
\setcounter{equation}{0}
\renewcommand{\thefigure}{A\arabic{figure}}
\setcounter{figure}{0}

\section{}

\subsection{Two scenarios for determining $N(\M,A)$}

\begin{figure*}[ht]
	\begin{center}
	\begin{tabular}{ccc}
	\includegraphics{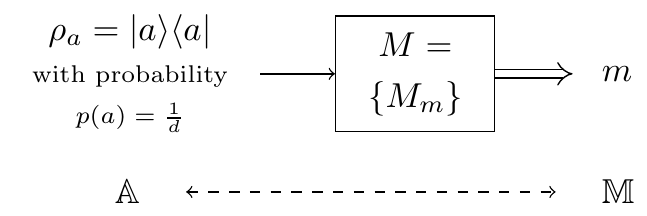}
	&
	\qquad\qquad\qquad
	&
	\includegraphics{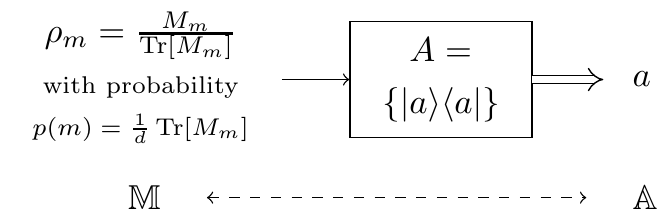}
	\\[2mm]
	\qquad \quad (a) && \qquad \qquad (b)\\
	\end{tabular}
	\end{center}
	\caption{As described in the text, the two situations above yield the same joint probability distribution $p(m,a) = \frac1d \Tr [M_m \oprod{a}{a}]$. The noise $N(\M,A)=H(\mathbb{A}|\mathbb{M})$ can readily be calculated in the second situation.}
	\label{fig:proof_NoiseConvEntropyRep}
\end{figure*}

The expression of the noise in terms of entropies of quantum observables in Eq.~\eqref{eqn:NoiseConvEntropyRep} shows that it is possible to determine the noise via two different experimental situations, both giving rise to the same joint probability distribution $p(m,a)$.

The first one is represented in Fig.~\ref{fig:proof_NoiseConvEntropyRep}(a), which is a simplified version of Figure~\ref{fig:NDSchematic}(a), in which the post-measurement state is ignored (recall indeed that it does not enter in the definition of the noise). 
The eigenstates $\ket{a}$ of $A$ are prepared with equal probabilities $p(a) = 1/d$, and measured by the POVM $\{M_m\}_m$. 
The second situation is that represented in Fig.~\ref{fig:proof_NoiseConvEntropyRep}(b): here, a quantum state $\rho_m = \frac{M_m}{\Tr[M_m]}$ is prepared with probability $p(m) = \frac1d \Tr[M_m]$, and undergoes a measurement of $A$. 
Eqs.~\eqref{eqn:mMarginal}--\eqref{eqn:NoiseConvEntropyRep} make it clear that both of these operational scenarios give rise to the same joint distributions $p(m,a)$, and thus both can equally well be used to determine the noise $N(\M,A)$.

Note that one way to prepare the states $\rho_m$ with the desired probabilities in the second situation is to measure the POVM $M^T = \{M_m^T\}_m$ on one subsystem of a pair in a maximally entangled state $\ket{\Phi^+} = \frac1{\sqrt{d}}\sum_j \ket{j}\ket{j}$ (where $\{\ket{j}\}$ denotes an orthonormal basis of the $d$-dimensional Hilbert space of one system, and ${\cdot}^T$ is the transposition in that basis). 
The same probability distribution $p(m,a)$ is then obtained in yet another scenario, which now involves the preparation of a fixed maximally entangled bipartite state, and measurements on both subsystems. 
This is precisely the scenario considered in the supplemental materials of Refs.~\cite{Buscemi:2014aa} and~\cite{Sulyok:2015aa} to calculate the noise $N(\M,A)$. 
Our derivation above shows that the introduction of an entangled state and the transpositions in those calculations were actually not necessary.

\subsection{Characterising the noise-noise region $R_{NN}(A,B)$}

The characterisation of the noise-noise region as in Proposition~\ref{prop:RNNDef} immediately follows from Eq.~\eqref{eqn:NoiseConvEntropyRep} along with the observation that the weighted ensemble of states $\{p(m),\rho_m\}_m$ defined above satisfies $\sum_m p(m) \rho_m=\id/d$, and that, vice versa, any weighted ensemble $\{p(m),\rho_m\}_m$ with $\sum_m p(m) \rho_m=\id/d$ defines a valid POVM $M=\{M_m = d \, p(m) \, \rho_m\}_m$.

As it turns out, the constraint $\sum_m p(m) \rho_m=\id/d$ can actually be disregarded in Eq.~\eqref{eqn:NNRegionForm} for the case of qubits. To see this, let $\{p(m),\,\rho_m\}_m$ be any arbitrary weighted ensemble of qubit states, i.e., $p(m)\ge 0$, $\sum_m p(m)=1$ and $\rho_m=\frac{1}{2}(\id+\vect{r}_m\cdot\vect{\sigma})$ where $|\vect{r_m}|\le 1$.
Then define $\{p(m,\pm),\,\rho_m^\pm\}_{m,\pm}$ with $p(m,\pm)=\frac{1}{2}p(m)$ and $\rho_m^\pm=\frac{1}{2}(\id\pm\vect{r}_m\cdot\vect{\sigma})$. This new ensemble satisfies
	\begin{align}
		\sum_{m,\pm}p(m,\pm) \rho_m^\pm = \sum_m\frac{1}{2}p(m)(\rho_m^+ + \rho_m^-)=\id/2.
	\end{align}
Furthermore, one has $H(A|\rho_m^\pm) = H(A|\rho_m)$ and similarly $H(B|\rho_m^\pm) = H(B|\rho_m)$, so that
	\begin{align}
		& \sum_{m,\pm} p(m,\pm) \big(H(A|\rho^\pm_m),\,H(B|\rho^\pm_m)\big) \notag \\[-2mm]
		& \qquad= \sum_m p(m) \big(H(A|\rho_m),\,H(B|\rho_m)\big).
	\end{align}
Hence, the ensemble $\{p(m),\,\rho_m\}_m$, which does not necessarily satisfy the constraint $\sum_m p(m) \rho_m=\id/d$, yields the same noise-noise values as another ensemble, which does satisfy the constraint. This proves that this constraint could indeed be removed from~\eqref{eqn:NNRegionForm}, from which it follows that the noise-noise region is then simply the convex hull of the preparation uncertainty region $E(A,B)$, as expressed by Proposition~\ref{prop:QubitRNNDef}.

Note that the above argument required considering a second ensemble with twice as many states as the original one -- or equivalently, due to the one-to-one correspondence highlighted above (for the second ensemble which does satisfy the previous normalisation constraint), a POVM with twice as many outcomes. 
Therefore the argument does not work if one imposes a fixed number of outcomes, as in the case of dichotomic measurements considered in the paper (for which the noise-noise region is then not necessarily convex).

\subsection{Relating the noise-noise and noise-disturbance regions}

Consider an arbitrary point in the noise-disturbance region $R_{ND}(A,B)$, obtained by some instrument $\M = \{\M_m\}_m$ and the optimal correction procedure $\E = \{\E_m\}_m$. We can combine $\M$, $\E$, and the final measurement of $B$ in Figure~\ref{fig:NDSchematic}(b) to define a global instrument $\M^{\E,B}$ (or a POVM $M^{\E,B}$, since the post-measurement state will not matter) with pairs of outcomes $(m,b')$.

The noises yielded by the instrument $\M^{\E,B}$ are then
	\begin{align}
		& N(\M^{\E,B},A) = H(\mathbb{A}|\mathbb{M},\mathbb{B}_{\M,\E}') \le H(\mathbb{A}|\mathbb{M}) = N(\M,A), \\
		& N(\M^{\E,B},B) = H(\mathbb{B}|\mathbb{M},\mathbb{B}_{\M,\E}') \le H(\mathbb{B}|\mathbb{B}_{\M,\E}') = D(\M,B),\label{eqn:DistBoundJointMB}
	\end{align}
where we have used the classical data-processing inequalities.

Hence, the noise $N(\M,A)$ and disturbance $D(\M,B)$ are bounded below by the noise values corresponding to another instrument $\M^{\E,B}$, which gives a point in the noise-noise region $R_{NN}(A,B)$. This proves Proposition~\ref{prop:NN-ND-relation}, that $R_{ND}(A,B)\subseteq \cl R_{NN}(A,B)$.

Note that the above argument does not hold if one imposes a limit on the number of outcomes, since the POVM $M^{\E,B}$ has $d\cdot|M|$ outcomes, where $d$ is the Hilbert space dimension and $|M|$ is the number of outcomes for $\M$.
The example given in Sec.~\ref{sec:NDdichotomic} of the main text for dichotomic measurements shows that one may, in such cases, indeed have $R^*_{ND}(A,B)\not\subseteq \cl R^*_{NN}(A,B)$.

Nevertheless, a similar argument can be used to show that one does have $R_{ND}^{MP}(A,B)\subseteq \cl R_{NN}^{MP}(A,B)$ when the measurements are performed by `measure-and-prepare' instruments (hence the superscript $MP$) if the number of outcomes is limited -- in particular, for dichotomic such measurements.
To see this, note as in~\cite{Sulyok:2016aa} that for such measurements, $\mathbb{B}\to\mathbb{M}\to\mathbb{B}_{\M,\E}'$ is a Markov chain and thus $H(\mathbb{B}|\mathbb{M},\mathbb{B}_{\M,\E}')=H(\mathbb{B}|\mathbb{M})=N(\M,B)$.
From Eq.~\eqref{eqn:DistBoundJointMB} we see that, for such measurements, $D(\M,B)\ge N(\M,B)$ and hence the noise and disturbance are bounded below by the noise values for the \emph{same} instrument (rather than the instrument $\M^{\E,B}$ used in the above, completely general, argument) proving the claim.

\subsection{Calculating the disturbance $D_\mathcal{I}(\M,B)$ for purity-preserving qubit instruments}

In order to derive an analytic formula for the disturbance $D_\mathcal{I}(\M,B)$ for purity-preserving qubit instruments we take a similar approach to that of Sec.~\ref{sec:NDLuders} for L\"{u}ders instruments, except now a further unitary transformation which depends on the measurement outcome is applied before measuring $B$. For simplicity we present here the calculation for $B=\sigma_x$ (as in Sec.~\ref{sec:NDLuders}), but it can straightforwardly be adapted to any Pauli observable $B$.

Let $M=\{M_m\}_m$ be the POVM corresponding to a purity-preserving instrument $\M=\{\M_m\}_m$.
Then, as discussed in Sec.~\ref{sec:characterisingRND}, on outcome $m$ the state is updated according to $\M_m(\rho)=U_m\sqrt{M_m}\rho\sqrt{M_m}U_m^\dagg$, where $U_m$ is a unitary transformation.
As in Eq.~\eqref{eqn:genPOVMform} we can write $M_m=p_m(\id+k_m\vect{n}_m\cdot\vect{\sigma})$, where $|\vect{n}_m|=1$, $p_m\ge 0$ and $|k_m| \le 1$, and which satisfies the normalisation constraints $\sum_m p_m=1$ and $\sum_m p_m k_m \vect{n}_m=\vect{0}$.

Calculating the post measurement states
\begin{equation}
	\rho_\pm=\M(\oprod{\pm x}{\pm x})=\frac{1}{2}(\id + \vect{r}_\pm\cdot\vect{\sigma})
\end{equation}
we find that $\vect{r}_\pm = \pm \vect{r}_0 + \vect{r}_\delta$, where
\begin{align}
	\vect{r}_0 = & \sum_m p_m  \left( (\vect{n}_m\cdot\vect{x})\vect{n}_m' \vphantom{\sqrt{1-k_m^2}}\right. \notag \\[-3mm]
	& \hspace{12mm} \left. + \sqrt{1-k_m^2}(\vect{x}_m'-(\vect{n}_m\cdot\vect{x})\vect{n}_m') \right), \\
	\vect{r}_\delta = & \sum_m p_mk_m\vect{n}_m',
\end{align}
and where $\vect{n}_m'$ and $\vect{x}_m'$ are rotations of $\vect{n}_m$ and $\vect{x}$ under $U_m$ satisfying $U_m(\vect{n}_m\cdot\vect{\sigma})U_m^\dagg=\vect{n}_m'\cdot\vect{\sigma}$ and $U_m(\vect{x}\cdot\vect{\sigma})U_m^\dagg=\vect{x}_m'\cdot\vect{\sigma}$.
Note that $\vect{r}_0$ can be obtained by rotating each summand in Eq.~\eqref{eqn:rvectorFull} for the Bloch vector obtained for L\"{u}ders instruments.
However, the presence of $\vect{r}_\delta$ means that, in stark contrast to the case for L\"{u}ders instruments, one generally has $\vect{r}_+\neq \vect{r}_-$.
The disturbance $D_\mathcal{I}(\M,\sigma_x)=H(\mathbb{X}|\mathbb{X}_{\M,\mathcal{I}}')$ can then be calculated directly to be
\begin{align}\label{eqn:DistPPgenForm}
	D_\mathcal{I}(\M,\sigma_x) = \sum_\pm \frac{1\pm \vect{r}_\delta\cdot\vect{x}}{2}\,h\left(\frac{|\vect{r}_0\cdot\vect{x}|}{1\pm\vect{r}_\delta\cdot\vect{x}}\right).
\end{align}

Note that if a single unitary $U$ is applied irrespective of the measurement outcome (i.e., $U_m=U$ for all $m$) one has $\vect{r}_\delta = \vect{0}$ and $\vect{r}_+=-\vect{r}_-$, as for L\"{u}ders instruments.
The disturbance $D_\mathcal{I}(\M,\sigma_x)$ is then simply $h(|\vect{r}_0\cdot\vect{x}|)$ and can be bounded below as in Eq.~\eqref{eqn:distBoundLudersInstr}.
One can then readily see that for such instruments the relation Eq.~\eqref{eqn:NDsqrtdynamics} is once again satisfied.

\subsection{Numerically sampling the points in $R_{ND}^{PP}(\sigma_z,\sigma_x)$}

In order to determine the lower boundary of the noise-disturbance region $R_{ND}^{PP}(\sigma_z,\sigma_x)$ for measurements with various numbers of outcomes, one wishes to sample instruments that are as close to this boundary as possible.
However, na\"{i}ve generation of random instruments performs very poorly at this.
In this section, we discuss some techniques for sampling large numbers of instruments that allow the lower boundary of $R_{ND}^{PP}(\sigma_z,\sigma_x)$ to be more easily investigated using numerical simulations.

Firstly, note that since both the noise~\eqref{eqn:NParameterised} and the disturbance~\eqref{eqn:DistPPgenForm} depend on the inner product of certain Bloch vectors with the $z$- and $x$-axes, one can essentially restrict oneself to this plane.
By considering POVMs and unitaries than act only in this plane one can sample more efficiently, and any $y$ component of the POVMs or post-measurement Bloch vectors serves only to increase both the noise and disturbance.%
\footnote{There are many ways one could generate random $k$-outcome POVMs to this end.
One such method would be to generate $k$ random states $\{\rho_m\}_m$ (e.g., with Bloch vectors uniformly distributed in the unit circle in the $xz$-plane) and a random probability distribution $\{p_m\}_m$ with $p_m\ge 0$ for all $m$ and $\sum_m p_m=1$ (e.g., by sampling from a Dirichlet distribution).
Let $\bar{\rho}=\sum_m p_m\rho_m$.
Then the operators $M_m=p_m(\sqrt{\bar{\rho}})^{-1}\rho_m(\sqrt{\bar{\rho}})^{-1}$ are Hermitian positive semidefinite and sum to the identity, and thus $\{M_m\}_m$ is a valid random POVM.
\label{fn:randomPOVM}}

In order to obtain initial bounds on $R_{ND}^{PP}(\sigma_z,\sigma_x)$, it is much more efficient to sample POVMs whose elements are all rank-one operators (i.e., proportional to projection observables).%
\footnote{The elements of any such POVM $\{M_m\}_m$ can be written $M_m=p_m(\id + \vect{n}_m\cdot\vect{\sigma})$ where $|\vect{n}_n|=1$, $p_m\ge 0$, $\sum_m p_m = 1$ and $\sum_m p_m \vect{n}_m = \vect{0}$. Random such POVMs with $k$ outcomes can easily be generated by choosing $k$ vectors $\vect{u}_m$ such that $\sum_m \vect{u}_m=\vect{0}$ and choosing $\vect{n}_m=\vect{u}_m/|\vect{u}_m|$ and $p_m=|\vect{u}_m|/(\sum_{m'}|\vect{u}_{m'}|)$.\label{fn:randomRank1POVM}}
These are simply the extremal POVMs~\cite{Haapasalo:2012aa}, although as a result of the apparent non-convexity of $R_{ND}^{PP}(\sigma_z,\sigma_x)$, these are not \emph{a priori} guaranteed to fully cover the noise-disturbance region.
However, empirically it does seem to be the case -- with the notable exception of the situation where the number of measurement outcomes is fixed -- that one obtains the same region whether or not one restricts oneself to such POVMs, and generally they provide data points much closer to the lower boundary of $R_{ND}^{PP}(\sigma_z,\sigma_x)$, thus allowing more efficient sampling.

Finally, for any given POVM $M=\{M_m\}_m$, one thus wishes to find the unitaries $\{U_m\}_m$ giving rise to the purity-preserving instrument that minimises Eq.~\eqref{eqn:DistPPgenForm}.
Although there seems to be no simple analytic approach to doing so, one can use numerical methods to perform such a minimisation and probe more precisely the boundary of $R_{ND}^{PP}(\sigma_z,\sigma_x)$.
Such minimisation can, in reality, be rather slow, but a rather good heuristic is to choose the unitaries that rotate the summands in Eq.~\eqref{eqn:rvectorFull} onto the positive $x$-axis.
In practise this gives results that are close to optimal -- and in many cases, such as for the example in Sec.~\ref{sec:counterexample}, demonstrably optimal -- and can be performed very quickly, allowing efficient sampling.

The particular results shown in Fig.~\ref{fig:NDbound} for three- and four-outcome instruments were obtained using random POVMs with rank-one elements (generated using the procedure described in Footnote~\ref{fn:randomRank1POVM}) and unitary corrections found with the heuristic optimisation described above.
Those in Fig.~\ref{fig:NDdichotomic} for dichotomic instruments were obtained by POVMs with one rank-one element and one rank-two element%
\footnote{Such a POVM can, for example, be efficiently generated by choosing a random unit vector $\vect{n}$ in the $xz$-plane and a random $p\in(0,1/2)$.
If we let $M_+=p(\id + \vect{n}\cdot\vect{\sigma})$ and $M_-=\id-M_+$ then $\{M_\pm\}_\pm$ is easily seen to be a random POVM with these properties.} 
(since no violation of Eq.~\eqref{eqn:NDrelnOrthog} appears possible with only rank-one elements) using a numerical optimisation for finding the best unitary corrections.

\bibliography{NoiseDisturbance}

\end{document}